\documentclass[12pt]{article}
\usepackage{amsmath}
\usepackage{amssymb}
\usepackage{amsfonts}
\usepackage{graphicx,psfrag,epsf}
\usepackage{natbib}
\usepackage[utf8]{inputenc}
 
\usepackage{hyperref}
\usepackage{mathtools}
\usepackage{threeparttable}
\usepackage{booktabs}
\usepackage{bookmark}
 
\usepackage{url} 
\usepackage{appendix}
 
\usepackage[figuresright]{rotating}

\usepackage{amsthm}
 
\def\bSig\mathbf{\Sigma}

\newcommand{\indep}{\perp \!\!\! \perp}
\DeclareMathOperator*{\argmax}{arg\,max}
\DeclareMathOperator*{\argmin}{arg\,min}

\newtheorem{theorem}{Theorem}
\newtheorem{lemma}{Lemma}

\begin{document}


  \title{\bf Estimating optimal individualized treatment rules with multistate processes}
  \author{Giorgos Bakoyannis \\
    Department of Biostatistics and Health Data Science \\
     Indiana University Fairbanks School of Public Health \\ 
     Indianapolis, IN 46202, U.S.A.}
  \date{}
  \maketitle

\def\spacingset#1{\renewcommand{\baselinestretch}%
{#1}\small\normalsize} \spacingset{1}


\begin{abstract}
Multistate process data are common in studies of chronic diseases such as cancer. These data are ideal for precision medicine purposes as they can be leveraged to improve more refined health outcomes, compared to standard survival outcomes, as well as incorporate patient preferences regarding quantity versus quality of life. However, there are currently no methods for the estimation of optimal individualized treatment rules with such data. In this article, we propose a nonparametric outcome weighted learning approach for this problem in randomized clinical trial settings. The theoretical properties of the proposed methods, including Fisher consistency and asymptotic normality of the estimated expected outcome under the estimated optimal individualized treatment rule, are rigorously established. A consistent closed-form variance estimator is provided and methodology for the calculation of simultaneous confidence intervals is proposed. Simulation studies show that the proposed methodology and inference procedures work well even with small sample sizes and high rates of right censoring. The methodology is illustrated using data from a randomized clinical trial on the treatment of metastatic squamous-cell carcinoma of the head and neck.
\end{abstract}

\noindent%
\vfill

\newpage
\spacingset{1.45} 

\section{Introduction}
\label{s:intro}

Modern precision medicine acknowledges the heterogeneity of the majority of human diseases and aims to develop and deliver therapies that are tailored to the individual patient. At the heart of these efforts is the development of data-driven individualized treatment assignment rules \citep{Kosorok19}. The purpose of such rules is to provide the right treatment to a given patient \citep{Lipkovich17} and, thereby, to improve health outcomes among patients overall. Over the past decade there has been a large number of new statistical methods for the estimation of optimal individualized treatment rules (ITRs) with various types of outcomes, including continuous, binary, survival, and competing risks outcomes. However, to the best of our knowledge, there are currently no methods for the estimation of optimal ITRs with stochastic processes that evolve through multiple discrete states over (continuous) time, also known as \textit{multistate processes}. Such processes are commonly encountered in studies of chronic diseases, such as cancer and HIV infection, where disease evolution is often characterized by multiple discrete health states. An example is the SPECTRUM trial \citep{Vermorken13}, a phase III randomized clinical trial on recurrent or metastatic squamous-cell carcinoma of the head and neck, where patient event history was characterized by the states ``tumor response" (i.e., tumor shrinkage per \citet{RECIST}), ``disease progression'', and ``death''. This paper addresses the issue of optimal ITR estimation with right-censored multistate processes data in randomized clinical trial settings. 

Multistate process data are ideal for precision medicine purposes as they can be leveraged to improve more refined health outcomes (compared to, e.g., overall or progression-free survival) as well as incorporate patient preferences. Typically, such data contain information about one or more transient health states (e.g., tumor response) that is not fully captured by standard survival or competing risks data. Therefore, they provide more comprehensive information about more refined health outcomes that may reflect both quantity and quality of life. Improvement of such outcomes may be more desirable than mere life extension to both patients and clinicians. In oncology, for example, such a desirable health outcome is sustained tumor response, which is associated with better quality of life, extended progression-free survival (PFS) time, and a prolonged treatment-free interval \citep{Kaufman17}. Tumor response was also a desirable outcome in the SPECTRUM trial. Another important outcome is extended quality-adjusted lifetime, which can be defined as the weighted sum of the times spent in a set of desirable health states \citep{Pelzer17, Oza20}. Quality-adjusted lifetime can be defined in multiple different ways according to patient preferences by using different weighting schemes that reflect the needs of different groups of patients. In this way, individual priorities towards quality versus quantity of life can be taken into account.

The methods for optimal ITR estimation can be classified into two broad categories: (i) backward induction methods such as Q-learning \citep{Murphy05, Zhao09} and A-learning \citep{Murphy03, Robins04} and (ii) direct-search methods, also known as policy-search methods, such as outcome weighted learning \citep{Zhao12} and value search estimation \citep{Zhang12}. The first class of methods typically estimates the optimal ITR by modeling either the conditional expectation of the outcome given the covariates (Q-learning) or the interactions between treatment and the covariates (A-learning). The second class of methods estimates the optimal ITR directly by optimizing an appropriate objective function. Given that imposing a realistic model for the conditional mean outcome or the treatment interactions is difficult for general multistate processes, the focus in this article is on direct-search methods. The issue of ITR estimation with survival outcomes has been well addressed in the literature. \citet{Zhao15} extended the outcome weighted learning framework of \citet{Zhao12} to situations where the outcome of interest is a right-censored survival time by proposing (i) an inverse censoring weighting (ICO) and (ii) a doubly robust (DR) outcome weighted learning approach. Estimation in these approaches relies on a weighted version of support vector machines \citep{Cortes95}. A similar method for ITR estimation with survival outcomes was proposed by \citet{Bai17}. A value search estimation approach for maximizing the $t$-year survival probability was developed by \citet{Jiang17}. This method approximates the nonsmooth objective function of the problem by a smooth one using kernel smoothing. Recently, there has also been some interest in the issue of ITR estimation with competing risk outcomes \citep{Zhou21,He21}. Nevertheless, the methods for survival and competing risk outcomes cannot be used for the estimation of optimal ITRs with right-censored multistate processes.

In this work, we extend the outcome weighted learning framework \citep{Zhao12,Zhao15} to deal with situations where the outcome of interest is an arbitrary right-censored multistate process, incorporating also patient preferences. The proposed method does not impose Markov assumptions or model assumptions on the multistate process of interest. The novelty of this paper is twofold. First, we devise a novel objective function which, in contrast to the ICO approach by \citet{Zhao15}, utilizes information from the censored cases. Importantly, this is achieved without imposing and estimating a model for the conditional expectation of the outcome given treatment and the covariates, as opposed to the DR approach by \citet{Zhao15}. Second, in addition to showing Fisher and universal consistency of the proposed method, we establish the asymptotic distribution of the proposed estimator for the expected outcome under the estimated optimal ITR and derive a consistent closed-form variance estimator. Based on our theoretical results, we also propose a method for the calculation of simultaneous confidence intervals over a set of patient preferences to account for multiplicity. The simulation studies provide evidence that the proposed estimator and inference procedures work well even with small sample sizes and under high rates of right censoring. Furthermore, the simulation studies reveal that the proposed method performs better than the previously proposed ICO and DR approaches for censored failure times.

\section{Methodology}
\label{s:estimation}

\subsection{Notation and data}
Consider a multistate process $\{X(t):t\in[0,\tau]\}$ with (finite) state space $\mathcal{S}=\{1,\ldots, S\}$, where $\tau<\infty$ is the maximum observation time. Let $Z\in\mathcal{Z}\subset \mathbb{R}^p$ denote a vector of baseline covariates that may be related to the effect of treatment on the multistate process of interest. For simplicity, the treatment $A$ is considered to be a binary variable taking its values in the treatment set $\{-1,1\}$. With multistate processes, the goal of treatment is typically to prolong the time spent in a set of desirable health states. Given that some health states may be more desirable than others and that patient preferences may vary, we define the benefit processes
\[
\left\{Y_w(t)=w'\tilde{X}(t):t\in[0,\tau],w\in\mathcal{W}\right\},
\]
where $\tilde{X}(t)=(I\{X(t)=1\},\ldots,I\{X(t)=S\})'$, $w$ is an $S$-dimensional vector of preference weights that satisfies $0\leq w\leq 1$ and the sum of its elements is positive and less than $S$, and $\mathcal{W}=\{w_1,\ldots,w_M\}$ is a prespecified finite set of preference weight vectors that reflect different patient preferences/priorities. Based on the latter processes, we define the utilities
\begin{equation}
\int_0^{\tau}Y_w(t)\textrm{d}m(t) \label{Yt}, \ \ \ \ w\in\mathcal{W},
\end{equation}
where the integrator function is $m(t)=t$ and induces the Lebesgue measure on the Borel $\sigma$-algebra $\mathcal{B}([0,\tau])$. Since $\int_0^{\tau}I\{X(t)=j\}\textrm{d}m(t)$ is the time spent in the $j$th state during the time interval $[0,\tau]$, the utilities \eqref{Yt} represent weighted sums of the (restricted) times spent in each state. For example, consider the states ``initial disease state'' (state 1), ``tumor response'' (state 2), and ``disease progression or death'' (state 3), in the setting of the SPECTRUM trial mentioned in the Introduction. A potential choice for the set of preference weights in this example is $\mathcal{W}=\{(0,1,0)',(0.5,1,0)',(1,1,0)'\}$. When $w=(0,1,0)'$, the utility $\int_0^{\tau}Y_w(t)\textrm{d}m(t)=\int_0^{\tau}I\{X(t)=2\}\textrm{d}m(t)$ corresponds to the restricted time spent in the tumor response state. The utility under the choice $w=(0.5,1,0)'$ corresponds to the restricted quality-adjusted lifetime $\int_0^{\tau}Y_w(t)\textrm{d}m(t)=0.5\times\int_0^{\tau}I\{X(t)=1\}\textrm{d}m(t) + \int_0^{\tau}I\{X(t)=2\}\textrm{d}m(t)$. This utility represents the restricted (quality-adjusted) time lived, where the time spent in the initial disease state (i.e., without tumor response) is reduced by 50\% to reflect a quality of life loss due to disease symptoms and/or side effects due to treatment continuation. When $w=(1,1,0)'$, the utility is the (restricted) PFS time, in which the times lived in the initial disease state and the tumor response state are equally important. Different choices of $w$ reflect different patient priorities towards quality versus quantity of life.

In practice, the multistate and benefit processes are not fully observed for all individuals due to the usual right censoring. Let $C$ denote the right censoring time and $T^*$ the time of arrival at an absorbing state (i.e., death). Letting $T=T^*\wedge\tau$, with $a\wedge b=\min(a,b)$ for any $a,b\in\mathbb{R}$, $\tilde{T}_i=T_i\wedge C_i$, and $\Delta_i=I(T_i\leq C_i)$, the observed data consist of independent and identically distributed copies of $
D_i=(Z_i,A_i,\tilde{T}_i,\Delta_i,\{X_i(t)I(C_i\geq T_i\wedge t):t\in[0,\tau]\})$, $i=1,\ldots,n$, where $X_i(t)I(C_i\geq T_i\wedge t)$, $t\in[0,\tau]$, is the censored version of the multistate process which is equal to 0 for the censored individuals after their censoring time. Note that $Y_{i,w}(t)=w'(I\{X_i(t)=1\},\ldots,I\{X_i(t)=S\})'$ is the benefit process for the $i$th individual.

\subsection{Optimal ITR estimation with multistate processes}
An ITR is a deterministic function $d:\mathcal{Z}\mapsto\{-1,1\}$ which suggests treatment choice $d(z)\in\{-1,1\}$ for a patient with covariates $Z=z$. Since the estimation of an optimal ITR is essentially a causal inference problem, we utilize the potential outcomes causal framework \citep{Rubin05,Tsiatis19}. Let $Y_w^*(t;1)$ and $Y_w^*(t;-1)$, $w\in\mathcal{W}$, denote the potential benefit processes under treatment choices 1 and $-1$, respectively, at time $t\in[0,\tau]$. Since the potential outcomes cannot be directly observed in real-world settings, we need to impose the following causal assumptions.
\begin{itemize}
\item[A1.] Stable unit treatment value assumption: $Y_w(\cdot)=Y_w^*(\cdot;1)I(A=1)+Y_w^*(\cdot;-1)I(A=-1)$. 
\item[A2.] Independent treatment assignment assumption: $\{Y_w^*(\cdot;1),Y_w^*(\cdot;-1),Z\}\indep A$.
\item[A3.] Positivity assumption: $\pi_0=P(A=1)\in[c_1,c_2]$, with $0<c_1<c_2<1$.
\end{itemize}
Assumptions A2 and A3 are automatically satisfied in randomized clinical trials. It must be highlighted that assumption A1 implies that the \textit{observed} benefit process $Y_w(\cdot)$, unlike the \textit{potential} benefit processes $Y_w^*(\cdot;1)$ and $Y_w^*(\cdot;-1)$, is associated with treatment $A$. Now, we can define the potential benefit processes under an ITR  $d$ as 
\[
Y_w^*(t;d)=Y_w^*(t;1)I\{d(Z)=1\} + Y_w^*(t;-1)I\{d(Z)=-1\}, \ \ \ \ w\in\mathcal{W}, \ \ t\in[0,\tau].
\] 
These processes are essential for defining potential benefit under an ITR for multistate processes. For such processes, we define the \textit{value} functions as
\[
\mathcal{V}_{w}(d)=E\left\{\int_0^{\tau}Y_w^*(t;d)\textrm{d}m(t)\right\}, \ \ \ \ w\in\mathcal{W},
\]
where $\mathcal{V}_w(d)$ is the expected sum of the weighted (under $w$) times spent in each state of the process during the time interval $[0,\tau]$, under the ITR $d$. Under assumptions A1--A3, the independent right censoring assumption, and the fact that
\[
\frac{Y_w(t)I(C\geq T\wedge t)}{\exp\{-\Lambda_0(T\wedge t)\}}=\frac{Y_w(t)I(C\geq T\wedge t)}{\exp\{-\Lambda_0(\tilde{T}\wedge t)\}}, \ \ \ \ t\in[0,\tau], \ \ w\in\mathcal{W},
\]
it can be shown that the value functions can be expressed in terms of the observable data as 
\[
\mathcal{V}_w(d)=E\left(\left[\int_0^{\tau}\frac{Y_{w}(t)I(C\geq T\wedge t)}{\exp\{-\Lambda_0(\tilde{T}\wedge t)\}}\textrm{d}m(t)\right]\frac{I(A=d(Z))}{A\pi_0 + (1-A)/2}\right), \ \ \ \ w\in\mathcal{W},
\]
where $\Lambda_0(t)$ is the cumulative hazard function of the right censoring variable $C$ at time $t$. An optimal ITR $d_w^*$, $w\in\mathcal{W}$, is a maximizer of the corresponding value function, i.e.
\[
d_w^*\in\underset{d}{\argmax}\mathcal{V}_w(d).
\]
Given that any rule $d:\mathcal{Z}\mapsto\{-1,1\}$ can be expressed as $d(z)=\textrm{sgn}\{f(z)\}$, where $\textrm{sgn}(x)=I(x\geq 0)-I(x<0)$, for some measurable function $f:\mathcal{Z}\mapsto\mathbb{R}$, and since
\begin{eqnarray*}
\mathcal{V}_w(d)&=&E\left(\left[\int_0^{\tau}\frac{Y_{w}(t)I(C\geq 
T\wedge t)}{\exp\{-\Lambda_0(\tilde{T}\wedge t)\}}\textrm{d}m(t)\right]\frac{1}{A\pi_0 + (1-A)/2}\right) \\
&&-E\left(\left[\int_0^{\tau}\frac{Y_{w}(t)I(C\geq T\wedge t)}{\exp\{-\Lambda_0(\tilde{T}\wedge t)\}}\textrm{d}m(t)\right]\frac{I(A\neq d(Z))}{A\pi_0 + (1-A)/2}\right),
\end{eqnarray*}
an optimal ITR is $d_w^*=\textrm{sgn}\{f_w^*(Z)\}$, where $f_w^*$ is a minimizer of the \textit{risk} function
\begin{eqnarray*}
\mathcal{R}_w(f)&=&E\left(\left[\int_0^{\tau}\frac{Y_{w}(t)I(C\geq T\wedge t)}{\exp\{-\Lambda_0(\tilde{T}\wedge t)\}}\textrm{d}m(t)\right]\frac{I(A\neq \textrm{sgn}(f(Z))}{A\pi_0 + (1-A)/2}\right) \\
&=&E\left(\left[\int_0^{\tau}\frac{Y_{w}(t)I(C\geq T\wedge t)}{\exp\{-\Lambda_0(\tilde{T}\wedge t)\}}\textrm{d}m(t)\right]\frac{I(Af(Z)<0)}{A\pi_0 + (1-A)/2}\right),
\end{eqnarray*}
over all measurable functions $f:\mathcal{Z}\mapsto\mathbb{R}$. It is not hard to see that 
\[
\mathcal{V}_w(\textrm{sgn}(f_w^*))-\mathcal{V}_w(\textrm{sgn}(f))=\mathcal{R}_w(f)-\mathcal{R}_w(f_w^*), \ \ \ \ w\in\mathcal{W},
\]
for any measurable $f$. Minimizing $\mathcal{R}_w(f)$ over all measurable functions is clearly not feasible. Therefore, we will consider minimization over a subset of the class of all measurable functions $f:\mathcal{Z}\mapsto\mathbb{R}$. In particular, we consider either of the following subsets.
\begin{itemize}
\item Class of linear functions $\{f(\cdot)=\beta_0+\langle\beta_1,\cdot\rangle:\beta_0\in \mathbb{R},\beta_1\in \mathbb{R}^p\}$, where $\langle \beta_1,z\rangle=\beta_1'z$.
\item Reproducing kernel Hilbert space (RKHS) $\mathcal{H}_k$ with kernel $k$, which is the completion of the space $\{f(\cdot)=\sum_{j=1}^m \alpha_jk(\cdot,z_j):m\in\mathbb{N},z_j\in\mathcal{Z},\alpha_j\in\mathbb{R}\}$. Here we consider the RKHS with the Gaussian kernel, $k(z_1,z_2)=\exp(-\sigma_n\|z_1-z_2\|^2)$, $z_1,z_2\in\mathcal{Z}$.
\end{itemize}
Minimizing the empirical version of $\mathcal{R}_w(f)$, over the chosen class $\mathcal{F}$, is challenging because it involves a discontinuous and nonconvex function of $f$. To alleviate, we follow the paradigm of outcome weighting learning \citep{Zhao12,Zhao15} and support vector machines \citep{Steinwart08} and utilize the surrogate risk
\[
\mathcal{R}_{\phi,w}(f)=E\left(\left[\int_0^{\tau}\frac{Y_{w}(t)I(C\geq T\wedge t)}{\exp\{-\Lambda_0(\tilde{T}\wedge t)\}}\textrm{d}m(t)\right]\frac{\phi(Af(Z))}{A\pi_0 + (1-A)/2}\right),
\]
where $\phi(x)=\max(0,1-x)$ is the hinge loss, which is convex in $f$. The cumulative hazard function $\Lambda_0$ of the right censoring distribution can be estimated using the nonparametric Nelson--Aalen estimator
\[
\hat{\Lambda}_n(t)=\int_0^t\frac{\sum_{i=1}^n\textrm{d}N_i(u)}{\sum_{i=1}^nY_{i}(u)}, \ \ \ \ t\in[0,\tau],
\]
where $N_i(t)=(1-\Delta_i)I(\tilde{T}_i\leq t)$ and $Y_i(t)=I(\tilde{T}_i\geq t)$. An obvious estimator of $\pi_0$ is $\hat{\pi}_n=n^{-1}\sum_{i=1}^nI(A_i=1)$. Thus, the empirical version of the surrogate risk $\mathcal{R}_{\phi,w}$ is
\begin{equation}
\hat{\mathcal{R}}_{\phi,w}(f)=\frac{1}{n}\sum_{i=1}^n\left(\left[\int_0^{\tau}\frac{Y_{i,w}(t)I(C_i\geq T_i\wedge t)}{\exp\{-\hat{\Lambda}_n(\tilde{T}_i\wedge t)\}}\textrm{d}m(t)\right]\frac{\phi(A_if(Z_i))}{A_i\hat{\pi}_n + (1-A_i)/2}\right). \label{proposed_R}
\end{equation}
Note that, even though $\pi_0$ is known in clinical trials, the estimate $\hat{\pi}_n$ is used in \eqref{proposed_R} as this typically leads to some efficiency gain in inverse probability weighting type estimators \citep{Tsiatis19}. The empirical surrogate risk for the ICO estimator for censored failure times (e.g., $\tilde{T}$) by \citet{Zhao15} for the setting considered here is
\[
\frac{1}{n}\sum_{i=1}^n\left(\left[\frac{\Delta_i\tilde{T}_i}{\exp\{-\hat{\Lambda}_n(\tilde{T}_i)\}}\right]\frac{\phi(A_if(Z_i))}{A_i\pi_0+(1-A_i)/2}\right).
\]
The latter incorporates inverse censoring weighting in the censored time of interest and discards the censored observations ($\Delta_i=0$). In contrast, the proposed empirical surrogate risk \eqref{proposed_R} utilizes information from both uncensored and censored cases by incorporating inverse censoring weighting in the underlying stochastic process $\{Y_w(t):t\in[0,\tau]\}$. Importantly, this is achieved without imposing and estimating a model for the conditional expectation of the time of interest given $A$ and $Z$, as opposed to the DR approach \citep{Zhao15}. Simulation studies summarized in Section~\ref{s:sims} reveal that these characteristics of the proposed method lead to a better performance compared to the ICO and DR estimators \citep{Zhao15} for censored failure times.

Similarly to \citet{Zhao12} and \citet{Zhao15}, the estimators of the optimal decision functions within the class $\mathcal{F}$ are obtained using penalized minimization of $\hat{\mathcal{R}}_{\phi,w}$ as 
\[
\hat{f}_{n,w,\lambda_n}=\argmin_{f\in\mathcal{F}}\left\{\hat{\mathcal{R}}_{\phi,w}(f)+\lambda_n\|f\|^2\right\}, \ \ \ \ w\in\mathcal{W},
\]
where $\lambda_n$ is a positive tuning parameter that controls the complexity of $f$ and $\|\cdot\|$ is a norm on the chosen class $\mathcal{F}$. For example, $\|\cdot\|$ is the Euclidean norm if $\mathcal{F}$ is the class of linear functions. For notational simplicity we omit the subscript $\lambda_n$ from the estimated decision functions and use the more compact notation $\hat{f}_{n,w}$. Based on the estimated decision functions $\{\hat{f}_{n,w}:w\in\mathcal{W}\}$ the estimated optimal ITRs are
\[
\hat{d}_{n,w}(z)=\textrm{sgn}\{\hat{f}_{n,w}(z)\}, \ \ \ \ w\in\mathcal{W}, \ \ z\in\mathcal{Z}.
\]
Based on the class of estimated ITRs $\{\hat{d}_{n,w}:w\in\mathcal{W}\}$, treatment assignment for a given patient utilizes $\hat{d}_{n,w}$, for the weight vector $w$ that is closest, with respect to the Euclidean norm, to the preference weight $w_0$ that reflects the patient's own preferences/priorities.

\subsection{Estimation of the value function}
The value function of an ITR $d$ can be estimated as
\[
\hat{\mathcal{V}}_{n,w}(d)=\frac{1}{n}\sum_{i=1}^n\left(\left[\int_0^{\tau}\frac{Y_{i,w}(t)I(C_i\geq T_i\wedge t)}{\exp\{-\hat{\Lambda}_n(\tilde{T}_i\wedge t)\}}\textrm{d}m(t)\right]\frac{I(A_i=d(Z_i))}{A_i\hat{\pi}_n + (1-A_i)/2}\right), \ \ \ \ w\in\mathcal{W}.
\]
The value $\hat{\mathcal{V}}_{n,w}(d)$ can be seen as the (estimated) performance of the ITR $d$ under the preference weight vector $w$. The estimated value of the estimated optimal ITR $\hat{\mathcal{V}}_{n,w}(\hat{d}_{n,w})$ is expected to be an over-optimistic estimate of the performance of $\hat{d}_{n,w}$ in a future (out of sample) patient when sample size is small relatively to $p$ or relatively to the complexity of $\mathcal{F}$. This is because the estimator $\hat{\mathcal{V}}_{n,w}(\hat{d}_{n,w})$ uses the same set of data $\{D_i: i=1,\ldots,n\}$, for both the estimation of the optimal ITR and the evaluation of its performance. This phenomenon is similar to the behaviour of the training error in support vector machines \citep{Hastie09}. A better estimate of $\mathcal{V}_w(\hat{d}_{n,w})$ in finite samples is expected to be the jackknife or leave-one-out cross-validation value estimator
\begin{equation}
\hat{\mathcal{V}}^{jk}_{n,w}(\hat{d}_{n,w})=\frac{1}{n}\sum_{i=1}^n\left(\left[\int_0^{\tau}\frac{Y_{i,w}(t)I(C_i\geq T_i\wedge t)}{\exp\{-\hat{\Lambda}_n(\tilde{T}_i\wedge t)\}}\textrm{d}m(t)\right]\frac{I(A_i=\hat{d}_{n,w}^{(-i)}(Z_i))}{A_i\hat{\pi}_n + (1-A_i)/2}\right), \ \ \ \ w\in\mathcal{W}, \label{V_jack}
\end{equation}
where $\hat{d}_{n,w}^{(-i)}=\textrm{sgn}(\hat{f}_{n,w}^{(-i)})$ is the optimal ITR estimated under the preference weight vector $w$ using all but the data of the $i$th individual.

\section{Theoretical Properties}
\label{s:properties}
The first theorem justifies the use of the surrogate risk $\mathcal{R}_{\phi,w}$, instead of the original risk $\mathcal{R}_w$, for the estimation of optimal ITR $d_w^*$, $w\in\mathcal{W}$.
\begin{theorem}[Fisher consistency]
Suppose that assumptions A1--A3 and condition C1 in Appendix A hold. Then, for any $w\in\mathcal{W}$, if $\tilde{f}_w$ minimizes $\mathcal{R}_{\phi,w}$, $d_w^*(z)=\textrm{sgn}\{\tilde{f}_w(z)\}$ for all $z\in\mathcal{Z}$.
\end{theorem}
The proof of Theorem 1 can be found in Appendix A.1. The next theorem ensures that $\mathcal{R}_{\phi,w}(\hat{f}_{n,w})$, which is the true surrogate risk of the estimated decision function $\hat{f}_{n,w}$, converges (in probability) to the minimal surrogate risk over the chosen class $\mathcal{F}$. It also asserts that, if the chosen class $\mathcal{F}$ is appropriate, then the proposed estimator $\hat{d}_{n,w}$ is universally consistent, i.e., its value converges (in probability) to the optimal value $\mathcal{V}_w(d_w^*)$. 
\begin{theorem}
Suppose that assumptions A1--A3 and conditions C1, C3, and C4 in Appendix A hold. Then, for $\lambda_n>0$ with $\lambda_n\rightarrow 0$, and $n\lambda_n\rightarrow\infty$, 
\[
\left|\mathcal{R}_{\phi,w}(\hat{f}_{n,w})-\inf_{f\in\mathcal{F}}\mathcal{R}_{\phi,w}(f)\right|\overset{p}{\rightarrow} 0, \ \ \ \ w\in\mathcal{W}, 
\]
as $n\rightarrow\infty$, for any distribution $P$ of the data $D$. Moreover, if (i) $\mathcal{F}$ is the space of linear functions and $f_w^*\in\mathcal{F}$ or (ii) $\mathcal{F}$ is the RKHS with the Gaussian kernel and the marginal distribution $\mu$ of $Z$ is regular, then
\[
\left|\mathcal{V}_w(\hat{d}_{n,w})-\mathcal{V}_w(d_w^*)\right|\overset{p}{\rightarrow} 0, \ \ \textrm{as} \ \ n\rightarrow\infty.
\]
\end{theorem}
Theorem 2 is proved in Appendix A.2. When $\mathcal{F}$ is the space of linear functions, universal consistency requires that the optimal decision function $f_w^*$ is linear, i.e. $f_w^*\in\mathcal{F}$. This requirement can be made more plausible by considering an enlarged covariate space $\tilde{\mathcal{Z}}$ that includes polynomial terms and/or two-way interaction terms between the original covariates $Z$. If $f_w^*\notin\mathcal{F}$, $\mathcal{V}_w(\hat{d}_{n,w})$ is expected to converge to a value that is lower than the optimal value $\mathcal{V}_w(d_w^*)$. Nevertheless, the limit of $\mathcal{V}_w(\hat{d}_{n,w})$ can be seen as an approximation to the optimal value $\mathcal{V}_w(d_w^*)$ by the first statement of Theorem 2 and the fact that $\mathcal{R}_w(f_w^*)\leq \mathcal{R}_w(f)\leq\mathcal{R}_{\phi,w}(f)$ for any $f\in\mathcal{F}$, since the hinge loss satisfies $\phi(x)\geq I(x<0)$ for all $x\in\mathbb{R}$. 

Interestingly, when $\mathcal{F}$ is the RKHS with the Gaussian kernel and the marginal distribution of $Z$ is regular, the estimated ITR is always universally consistent. However, the so-called no-free-lunch theorem \citep{Steinwart08} implies that the corresponding rate of convergence can be extremely slow for at least some distributions of the data. This means that an extremely large sample size may be required in practice in order to obtain an ITR $\hat{d}_{n,w}$ with a value reasonably close to the optimal value. For this reason, we will restrict our attention to the case where $\mathcal{F}$ is the space of linear functions in the remainder of this paper. 

The next theorem characterizes the asymptotic distribution of the estimated value function of a fixed decision function $f$. 
\begin{theorem}
Under assumptions A1--A3 and conditions C1, C2, and C4 in Appendix A, we have that
\[
G_{n,w}(f)\coloneqq\sqrt{n}\left\{\hat{\mathcal{V}}_{n,w}(\textrm{sgn}(f))-\mathcal{V}_w(\textrm{sgn}(f))\right\}=\frac{1}{\sqrt{n}}\sum_{i=1}^n\psi_{i,w}(f)+\epsilon_{n,w}(f), \ \ \ \ w\in\mathcal{W},
\]
for any given measurable decision function $f$, where the explicit formula for the influence function $\psi_{i,w}(f)$ is provided in Appendix B and $\epsilon_{n,w}(f)=o_p(1)$. Moreover, if $\mathcal{F}$ is the space of linear functions, then $\sup_{f\in\mathcal{F}}|\epsilon_{n,w}(f)|=o_p(1)$ and the class of influence functions $\{\psi_w(f):f\in\mathcal{F}\}$ is $P$-Donsker.
\end{theorem}
The proof of Theorem 3 is given in Appendix A.3. This theorem implies that, for any  measurable decision function $f$, $G_{n,w}(f)$ is asymptotically normal with mean zero and variance $\sigma_w^2(f)=E\psi_{1,w}^2(f)$. This asymptotic variance can be consistently (in probability) estimated by $\hat{\sigma}_w^2(f)=n^{-1}\sum_{i=1}^n\hat{\psi}_{i,w}^2(f)$, where $\hat{\psi}_{i,w}(f)$, $i=1,\ldots,n$, are the empirical versions of the influence functions. The latter are obtained by replacing expectations by sample averages and unknown parameters by their consistent estimates in $\psi_{i,w}(f)$. The formula for the empirical influence function is provided in Appendix B. 

The final theorem characterizes the asymptotic distribution of the estimated value function of $\hat{d}_{n,w}=\textrm{sgn}(\hat{f}_{n,w})$ and provides the basis for the calculation of simultaneous confidence intervals over the preference weight set $\mathcal{W}=\{w_1,\ldots,w_M\}$. 
\begin{theorem}
Suppose that $\mathcal{F}$ is the space of linear functions. Then, under assumptions A1--A3, conditions C1--C5 in Appendix A, the additional assumption that the optimal linear decision function $\tilde{f}_w(z)$ satisfies $P(\tilde{f}_w(Z)=0)=0$, $w\in\mathcal{W}$, and for $\lambda_n>0$ with $\lambda_n\rightarrow 0$ and $n\lambda_n\rightarrow\infty$, we have that
\[
\max_{w\in\mathcal{W}}\left|G_{n,w}(\hat{f}_{n,w})-G_{n,w}(\tilde{f}_w)\right|=o_p(1).
\]
\end{theorem}

The proof of Theorem 4 is given in Appendix A.4. The last theorem and Theorem 3 imply that $G_{n,w}(\hat{f}_{n,w})$ is asymptotically normal with zero mean and variance $\sigma_{w}^2(\tilde{f}_w)$. This variance can be consistently (in probability) estimated by $\hat{\sigma}_{w}^2(\hat{f}_{n,w})=n^{-1}\sum_{i=1}^n\hat{\psi}_{i,w}^2(\hat{f}_{n,w})$. This result can be used for the calculation of (pointwise) confidence intervals and conducting hypothesis testing regarding the performance of the estimated ITR $\mathcal{V}_w(\textrm{sgn}(\hat{f}_{n,w}))$ under the preference weight $w$. Theorems 3 and 4 can also be used for simultaneous inference. Define the vector $G_n\coloneqq(G_{n,w_1}(\hat{f}_{n,w_1}),\ldots,G_{n,w_M}(\hat{f}_{n,w_M}))'$ and let $\|x\|_{\infty}\coloneqq \max_{1\leq l\leq M}|x_l|$ denote the maximum norm of a vector $x=(x_1,\ldots,x_M)'\in\mathbb{R}^M$. Theorems 3 and 4, the Cram{\'e}r--Wold theorem, and the continuous mapping theorem, imply that, for any fixed matrix $Q$ (of appropriate dimension),
\[
\|QG_n\|_{\infty}\overset{d}\rightarrow \|QG\|_{\infty},
\]
where $G\sim N(0,\Omega)$, with $\Omega$ being a positive semidefinite matrix with elements $\omega_{l,l}=\sigma_{w_l}^2(\tilde{f}_{w_l})$, $l=1,\ldots,M$, and $\omega_{l,l'}=E\{\psi_{1,w_l}(\tilde{f}_{w_l})\psi_{1,w_{l'}}(\tilde{f}_{w_{l'}})\}$, $l\neq l'$. The covariance matrix $\Omega$ can be consistently (in probability) estimated by the matrix $\hat{\Omega}_n$ with elements $\hat{\omega}_{l,l}=\hat{\sigma}_{w_l}^2(\hat{f}_{n,w})$, $l=1,\ldots,M$, and  $\hat{\omega}_{l,l'}=n^{-1}\sum_{i=1}^n\hat{\psi}_{i,w_l}(\hat{f}_{n,w_l})\psi_{i,w_{l'}}(\hat{f}_{n,w_{l'}})$, $l\neq l'$. Setting $Q=\textrm{diag}(\sigma_{w_1}^{-1}(\tilde{f}_{w_1}),\ldots,\sigma_{w_M}^{-1}(\tilde{f}_{w_M}))$ implies that
\[
\max_{w\in\mathcal{W}}|\sigma_{w}^{-1}(\tilde{f}_w)G_{n,w}(\hat{f}_{n,w})|\overset{d}\rightarrow\|QG\|_{\infty}.
\]
The last result implies that $1-\alpha$ simultaneous confidence intervals over $\mathcal{W}$, which account for the multiplicity due to considering multiple preference weights, can be calculated as 
\[
\left[\hat{\mathcal{V}}_{n,w}(\textrm{sgn}(\hat{f}_{n,w}))-c_\alpha\frac{\hat{\sigma}_w(\hat{f}_{n,w})}{\sqrt{n}},\hat{\mathcal{V}}_{n,w}(\textrm{sgn}(\hat{f}_{n,w}))+c_\alpha\frac{\hat{\sigma}_w(\hat{f}_{n,w})}{\sqrt{n}}\right], \ \ \ \ w\in\mathcal{W},
\]
where $c_a$ is the $1-\alpha$ percentile of the distribution of $\|QG\|_{\infty}$. This percentile can be easily estimated using the following  simulation procedure. First, choose a large number $B$ (say $B=1000$) and simulate vectors $G_b\sim N(0,\hat{\Omega}_n)$, for $b=1,\ldots,B$. Then, $c_{\alpha}$ can be estimated as the empirical $1-\alpha$ percentile of the sample $\|\hat{Q}_nG_1\|_{\infty},\ldots,\|\hat{Q}_nG_B\|_{\infty}$, where $\hat{Q}_n=\textrm{diag}(\hat{\sigma}_{w_1}^{-1}(\hat{f}_{n,w_1}),\ldots,\hat{\sigma}_{w_M}^{-1}(\hat{f}_{n,w_M}))$. Simultaneous confidence intervals for the differences $\mathcal{V}_{w}(\textrm{sgn}(\hat{f}_{n,w}))-\mathcal{V}_{w}(1)$ and $\mathcal{V}_{w}(\textrm{sgn}(\hat{f}_{n,w}))-\mathcal{V}_{w}(-1)$, $w\in\mathcal{W}$, where $\mathcal{V}_{w}(1)$ and $\mathcal{V}_{w}(-1)$ are the value functions for the fixed rules $d(z)\coloneqq 1$ and $d(z)\coloneqq -1$, can be calculated similarly by expanding the vector $G_n$ to include $G_{n,w}(1)$ and $G_{n,w}(-1)$, $w\in\mathcal{W}$, and using an appropriate matrix $Q$. This is illustrated in the data application presented in Section~\ref{s:analysis}. We argue that the same inference procedures can also be used for the jackknife value estimator $\hat{\mathcal{V}}^{jk}_{n,w}(\textrm{sgn}(\hat{f}_{n,w}))$, $w\in\mathcal{W}$. This statement is justified numerically in the simulation studies.

\section{Simulation Studies}
\label{s:sims}
The finite sample performance of the proposed methods was evaluated in a series of simulation experiments. Specifically, we assessed the performance of (i) the proposed ITR estimator $\hat{d}_{n,w}$ and (ii) the proposed inference methods for the (true) value of the estimated ITR $\mathcal{V}_w(\hat{d}_{n,w})$. We considered a binary treatment variable $A\in\{-1,1\}$, a two-dimensional covariate vector $Z=(Z_1,Z_2)'$, and a multistate process $\{X(t):t\in[0,\tau]\}$ under a progressive illness-death model with state space $\mathcal{S}=\{1,2,3\}$. State 1 represented the initial disease state, state 2 the tumor response state, and state 3 the disease progression or death state, in a hypothetical oncology trial. We considered the preference weights $w_1=(0,1,0)'$ and $w_2=(1,1,0)'$, which correspond to the (restricted) mean duration of tumor response and the PFS time, respectively. Treatment was simulated with $P(A=1)=0.5$, while the covariates were simulated from the uniform distribution $U(-1,1)$. The multistate process was simulated, conditionally on $A$ and $Z$, based on the transition intensities $\alpha_{12}(A,Z)=\exp(-0.5Z_1+0.5Z_2+Af_w^*(Z))$, $\alpha_{13}(A,Z)=\exp(-0.5Z_1+0.5Z_2)/4$, and $\alpha_{23}(A,Z)=\exp(-0.5Z_1+0.5Z_2-Af_w^*(Z))/2$, where $\alpha_{hj}(a,z)$ represents the transition rate from state $h$ to state $j$ for a patient with $(A,Z')'=(a,z')'$ and $f_w^*$, $w\in\{w_1,w_2\}$, is the optimal decision function under $w$. Under the aforementioned choices, $f_{w_1}^*=f_{w_2}^*\equiv f_{w}^*$. In total, we considered four scenarios according to the form of the optimal decision function as follows:
\begin{itemize}
    \item Scenario 1: $f_w^*(Z) = Z_1 + Z_2$ 
    \item Scenario 2: $f_w^*(Z) = 2(Z_1 - Z_2)$
    \item Scenario 3: $f_w^*(Z) = 1 + Z_2 - \exp(-Z_1)$ 
    \item Scenario 4: $f_w^*(Z) = 2\log(2 - Z_1 - Z_2) - 1.4$
\end{itemize}
The right censoring time was simulated independently of the multistate process from the exponential distribution $\textrm{Exp}(\theta)$, with $\theta\in\{-1.6, -1, -0.4\}$, and the total duration of the study was set to $\tau=3$. These choices for $\theta$ and $\tau$ led on average to 28.4\%, 42.8\%, and 59.5\% right-censored observations, respectively.

For each scenario and censoring rate $\theta$, we considered the training sample sizes $n\in\{100,200,300,400\}$ and repeated the simulation 1000 times. In each training data set we applied the proposed method with the search space $\mathcal{F}$ being the space of linear functions and $\tau=3$. Note that, $f_w^*\notin\mathcal{F}$ in scenarios 3 and 4. The tuning parameter was set to $\lambda_n=n^{-1/2}$, which satisfies the requirements of Theorems 2 and 4. To evaluate the performance of the estimated ITR $\hat{d}_{n,w}$, we considered two metrics: (i) the estimated ITR value ratio $\mathcal{V}_w(\hat{d}_{n,w})/\mathcal{V}_w(d_w^*)$, where $d_w^*=\textrm{sgn}(f_w^*)$ and $\mathcal{V}_w(d_w^*)$ is the maximum possible value, and (ii) the misclassification rate, defined as the proportion of patients which were assigned by $\hat{d}_{n,w}$ to the wrong treatment. An estimated ITR value ratio close to 1 indicates that the performance of $\hat{d}_{n,w}$ is close to optimal. For each simulated training data set, the true value of the estimated ITR $\mathcal{V}_w(\hat{d}_{n,w})$ and the misclassification rate were calculated based on an independently simulated large testing data set of size 10000. To evaluate the validity of the proposed inference methods for $\mathcal{V}_w(\hat{d}_{n,w})$ we considered: (i) the average percent errors of the value function estimators $\hat{\mathcal{V}}_{n,w}(\hat{d}_{n,w})$ and $\hat{\mathcal{V}}_{n,w}^{jk}(\hat{d}_{n,w})$, defined as 
\[
\frac{1}{1000}\sum_{b=1}^{1000}\frac{\hat{\mathcal{V}}_{n,w,b}(\hat{d}_{n,w,b})-\mathcal{V}_w(\hat{d}_{n,w,b})}{\mathcal{V}_w(\hat{d}_{n,w,b})}\times 100 \ \ \textrm{and} \ \ \frac{1}{1000}\sum_{b=1}^{1000}\frac{\hat{\mathcal{V}}_{n,w,b}^{jk}(\hat{d}_{n,w,b})-\mathcal{V}_w(\hat{d}_{n,w,b})}{\mathcal{V}_w(\hat{d}_{n,w,b})}\times 100,
\]
where $\hat{\mathcal{V}}_{n,w,b}$, $\hat{\mathcal{V}}_{n,w,b}^{jk}$, and $\hat{d}_{n,w,b}$ are estimates from the $b$th simulated training data set, (ii) the average of the proposed standard error estimates relatively to the Monte Carlo standard deviation of the value estimates, and (iii) the coverage probability of the 95\% confidence intervals calculated using the proposed standard error estimator under asymptotic normality. 

The simulation results regarding the performance of the estimated ITR $\hat{d}_{n,w}$ for the time in response (i.e., $w=(0,1,0)'$) are depicted in Figure~\ref{f:sims}. The estimated ITR value ratio was above 0.9 in all cases. Thus, even in Scenarios 3 and 4 where $f_w^*$ is not linear, the performance of the estimated rule was close to optimal. The maximum fixed rule value ratio $\max\{\mathcal{V}_w(1),\mathcal{V}_w(-1)\}/\mathcal{V}_w(d_w^*)$, where $\mathcal{V}_w(1)$ and $\mathcal{V}_w(-1)$ are the values for the fixed rules $d(z)\coloneqq 1$ and $d(z)\coloneqq -1$, was close to 0.8 in all cases. This indicates that the estimated ITR $\hat{d}_{n,w}$ can lead to substantially better health outcomes on average compared to fixed, one-size-fits-all, rules. As expected, the estimated ITR value ratio was higher with larger training sample sizes and lower censoring rates. A similar pattern was observed for the misclassification rate of $\hat{d}_{n,w}$, which was lower for larger training samples and lower censoring rates. The simulation results regarding the validity of the proposed inference methods for $\mathcal{V}_w(\hat{d}_{n,w})$ are summarized in Tables~\ref{t:sims_lin} and \ref{t:sims_nonlin}. In all cases, the value estimator $\hat{\mathcal{V}}_{n,w}(\hat{d}_{n,w})$ provided slightly optimistic estimates of the true value $\mathcal{V}_w(\hat{d}_{n,w})$. The percent error was over 4\% only in a few cases with $n=100$ and a 60\% censoring rate. As expected, the percent error of the jackknife value estimator $\hat{\mathcal{V}}_{n,w}^{jk}(\hat{d}_{n,w})$ was lower than that of $\hat{\mathcal{V}}_{n,w}(\hat{d}_{n,w})$. However, the difference between the two value estimators was smaller for larger training samples and lower censoring rates. The average standard error estimates were close to the corresponding Monte Carlo standard deviation of the estimates and the coverage probabilities close to the nominal level in all cases. These results indicate the consistency of the proposed standard error estimator and support the asymptotic normality result from Theorems 3 and 4. 

\begin{figure}
\centerline{\includegraphics[width=6.8in]{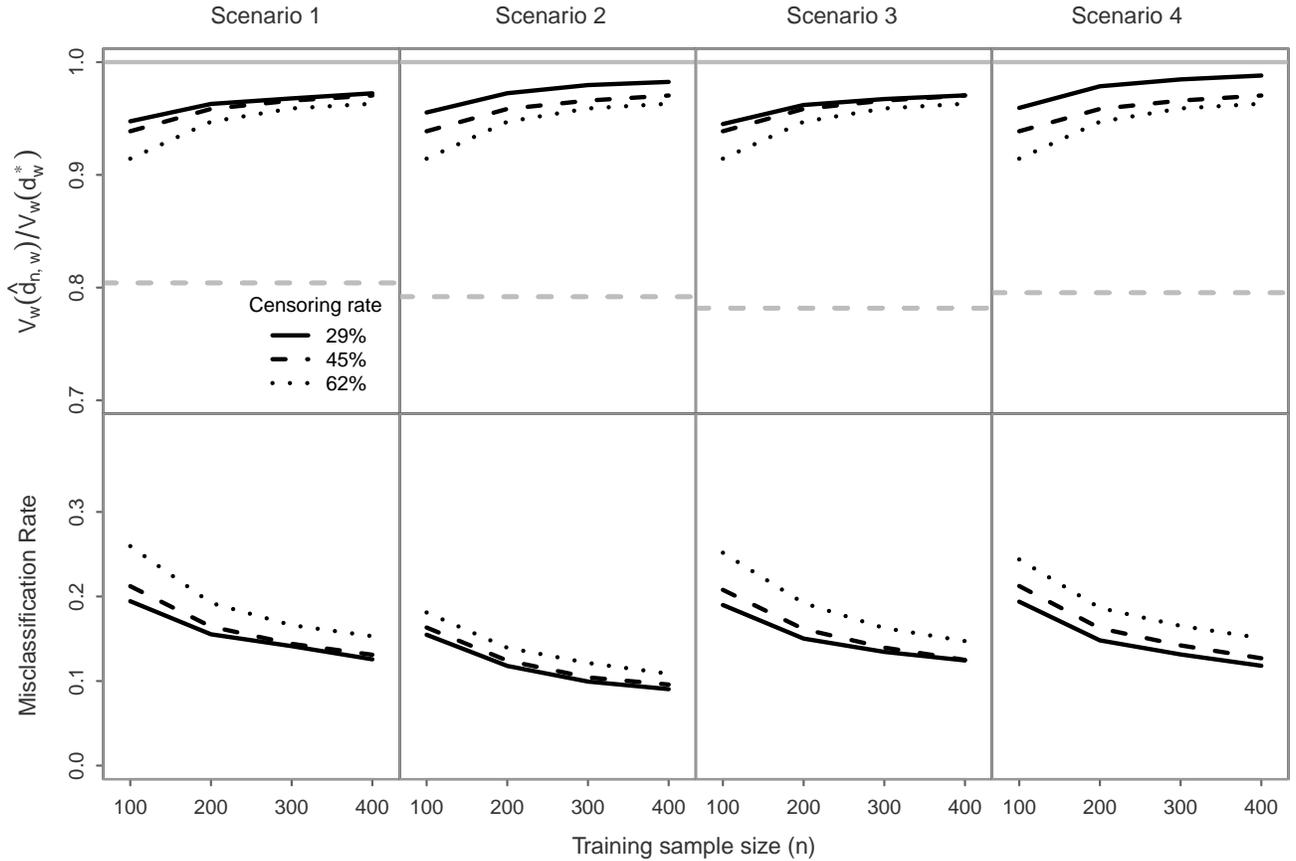}}
\caption{Simulation study: Performance of $\hat{d}_{n,w}$ for the duration of tumor response (i.e., $w=(0,1,0)'$, in terms of the average estimated ITR value ratio $\mathcal{V}_w(\hat{d}_{n,w})/\mathcal{V}_w(d_w^*)$ (top row) and misclassification rate (bottom row). The solid grey horizontal lines (top row) correspond to an optimal performance while the dashed grey horizontal lines (top row) correspond to the maximum fixed rule value ratio $\max\{\mathcal{V}_w(1),\mathcal{V}_w(-1)\}/\mathcal{V}_w(d_w^*)$, where $\mathcal{V}_w(1)$ and $\mathcal{V}_w(-1)$ are the values for the fixed rules $d(z)\coloneqq 1$ and $d(z)\coloneqq -1$.}
\label{f:sims}
\end{figure}

\begin{table}
\caption{Simulation study: Performance of the proposed inference methods for the true value of the estimated ITR $\mathcal{V}_w(\hat{d}_{n,w})$ for the duration of tumor response (i.e., $w=(0,1,0)'$), under a linear optimal decision function $f_w^*$ (scenarios 1 and 2). (Cens: right censoring rate; $n$: training sample size; MCSD: Monte Carlo standard deviation of the estimates; ASE: Average of the standard error estimates; CP: empirical coverage probability of the 95\% confidence interval)}
\label{t:sims_lin}
\begin{center}
\begin{tabular}{lllcccccccc}
\hline
&&& \multicolumn{4}{c}{$\hat{\mathcal{V}}_{n,w}(\hat{d}_{n,w})$} & \multicolumn{4}{c}{$\hat{\mathcal{V}}_{n,w}^{jk}(\hat{d}_{n,w})$} \\
\cmidrule(lr){4-7}\cmidrule(lr){8-11}
Scenario & Cens & $n$ &
\multicolumn{1}{c}{\% error} &
\multicolumn{1}{c}{MCSD} & 
\multicolumn{1}{c}{ASE} & 
\multicolumn{1}{c}{CP} &
\multicolumn{1}{c}{\% error} &
\multicolumn{1}{c}{MCSD} & 
\multicolumn{1}{c}{ASE} & 
\multicolumn{1}{c}{CP} \\ \hline
1&  28\%&100&3.418&0.224&0.215&0.954&-1.402&0.226&0.215&0.940 \\
&&200&2.129&0.151&0.155&0.961&-0.540&0.153&0.155&0.953 \\
&&400&1.062&0.109&0.110&0.954&-0.511&0.112&0.110&0.936 \\[1ex]
&  43\%&100&4.707&0.249&0.240&0.955&-0.856&0.248&0.240&0.946 \\
&&200&2.731&0.171&0.173&0.951&-0.039&0.172&0.173&0.949 \\
&&400&1.330&0.117&0.123&0.966&-0.630&0.120&0.123&0.954 \\[1ex]
&  60\%&100&5.289&0.313&0.303&0.951&-0.602&0.301&0.303&0.951 \\
&&200&3.411&0.215&0.218&0.964&0.016&0.218&0.218&0.947 \\
&&400&1.983&0.141&0.154&0.964&-0.681&0.143&0.154&0.962 \\[1.5ex]
2&  28\%&100&2.393&0.265&0.231&0.947&-0.663&0.261&0.231&0.940 \\
&&200&1.359&0.172&0.166&0.953&-0.529&0.174&0.166&0.941 \\
&&400&0.807&0.123&0.118&0.964&-0.218&0.124&0.118&0.953 \\[1ex]
&  42\%&100&2.403&0.288&0.258&0.951&-1.570&0.286&0.258&0.939 \\
&&200&1.568&0.194&0.185&0.953&-0.692&0.198&0.185&0.944 \\
&&400&1.207&0.139&0.132&0.946&-0.078&0.140&0.132&0.945 \\[1ex]
&  57\%&100&1.164&0.353&0.319&0.935&-3.364&0.343&0.319&0.931 \\
&&200&1.953&0.242&0.230&0.953&-0.785&0.236&0.230&0.944 \\
&&400&1.217&0.162&0.164&0.964&-0.492&0.165&0.164&0.955 \\
\hline
\end{tabular}
\end{center}
\end{table}

\begin{table}
\caption{Simulation study: Performance of the proposed inference methods for the true value of the estimated individualized treatment rule $\mathcal{V}_w(\hat{d}_{n,w})$ for the duration of tumor response (i.e., $w=(0,1,0)'$), under a nonlinear optimal decision function $f_w^*$ (scenarios 3 and 4). (Cens: right censoring rate; $n$: training sample size; MCSD: Monte Carlo standard deviation of the estimates; ASE: Average of the standard error estimates; CP: empirical coverage probability of the 95\% confidence interval)}
\label{t:sims_nonlin}
\begin{center}
\begin{tabular}{lllcccccccc}
\hline
&&& \multicolumn{4}{c}{$\hat{\mathcal{V}}_{n,w}(\hat{d}_{n,w})$} & \multicolumn{4}{c}{$\hat{\mathcal{V}}_{n,w}^{jk}(\hat{d}_{n,w})$} \\
\cmidrule(lr){4-7}\cmidrule(lr){8-11}
Scenario & Cens & $n$ &
\multicolumn{1}{c}{\% error} &
\multicolumn{1}{c}{MCSD} & 
\multicolumn{1}{c}{ASE} & 
\multicolumn{1}{c}{CP} &
\multicolumn{1}{c}{\% error} &
\multicolumn{1}{c}{MCSD} & 
\multicolumn{1}{c}{ASE} & 
\multicolumn{1}{c}{CP} \\ \hline
3&  29\%&100&3.706&0.240&0.216&0.958&-0.415&0.236&0.216&0.942 \\
&&200&2.242&0.160&0.156&0.963&-0.049&0.156&0.156&0.956 \\
&&400&1.136&0.113&0.111&0.956&-0.331&0.115&0.111&0.943 \\[1ex]
&  43\%&100&4.483&0.258&0.243&0.959&-0.999&0.265&0.243&0.935 \\
&&200&2.355&0.182&0.174&0.948&-0.563&0.184&0.174&0.937 \\
&&400&1.505&0.120&0.124&0.951&-0.264&0.124&0.124&0.949 \\[1ex]
&  60\%&100&4.719&0.317&0.305&0.949&-1.154&0.306&0.305&0.938 \\
&&200&3.421&0.220&0.219&0.962&0.254&0.214&0.219&0.957 \\
&&400&1.397&0.146&0.155&0.964&-0.967&0.143&0.155&0.961 \\[1.5ex]
4&  29\%&100&3.957&0.229&0.216&0.947&-0.411&0.234&0.216&0.945 \\
&&200&2.232&0.158&0.155&0.950&-0.269&0.160&0.155&0.944 \\
&&400&1.128&0.113&0.111&0.948&-0.380&0.115&0.111&0.936 \\[1ex]
&  43\%&100&4.419&0.262&0.240&0.947&0.260&0.268&0.240&0.927 \\
&&200&2.278&0.180&0.173&0.961&-0.786&0.181&0.173&0.948 \\
&&400&1.197&0.128&0.123&0.943&-0.648&0.129&0.123&0.932 \\[1ex]
&  59\%&100&5.189&0.303&0.306&0.962&-1.175&0.313&0.306&0.927 \\
&&200&2.862&0.224&0.218&0.952&-1.089&0.226&0.218&0.929 \\
&&400&1.526&0.161&0.154&0.949&-0.809&0.161&0.154&0.931 \\
\hline
\end{tabular}
\end{center}
\end{table}

Additional simulation results evaluating the effect of selecting different values of $\tau$ for the analysis are presented in Figures~1-3 in Appendix C.1. In these simulation studies we considered the values $\tau\in\{1,2,3\}$, with 3 being equal to the length of the study. In these studies, larger values of $\tau$ led to a better performance in terms of both the median value function and the variability. This is to be expected as more information is incorporated in the proposed method with a larger value of $\tau$. However, the differences between the choices $\tau=2$ and $\tau=3$ were not pronounced in general. Further simulation results on the performance of the proposed ITR estimator when $\mathcal{F}$ is the RKHS with the Gaussian kernel with $\sigma=1$ (less flexible kernel) and $\sigma=5$ (more flexible kernel), for the duration of tumor response, are presented in Figures~4-6 in Appendix C.1. A larger training sample $n$ led to a better performance in all cases, which reflects the consistency of the proposed ITR estimator. Using a more flexible kernel led to an inferior performance with smaller training sample sizes $n$. In most cases with $n=200$ or $n=400$, the performance of a less flexible kernel was the best, while the use of a linear decision function led in general to a better performance compared to two kernel choices when $n=100$. However, in many cases, the differences between a less flexible kernel and the linear decision function were not pronounced.

Simulation results regarding the performance of the estimated ITR $\hat{d}_{n,w}$ for the PFS time (i.e., $w=(1,1,0)'$) are depicted in Figures~7-9 in Appendix C.2. For comparison, these plots also illustrate the performance of the ICO and DR methods \citep{Zhao15}. To apply the DR approach, we estimated $E(T|T>t, A, Z)$ based on the semiparametric Cox model for $T$ with $A$, $Z$, and the interactions between $A$ and $Z$ as covariates, according to \citet{Zhao15}. This model is misspecified due to the complexity of the multistate process considered in the simulation studies. The performance of all methods improved with a larger sample size. Furthermore, the proposed method provided estimated ITRs with a substantially lower variability and slightly larger median value (i.e., PFS time) compared to the ICO and DR methods. The improved performance of the proposed method over the ICO and DR methods was more pronounced with a higher censoring rate. The simulation results regarding the validity of the proposed inference methods for $\mathcal{V}_w(\hat{d}_{n,w})$ are summarized in Tables~1 and 2 in Appendix C.2. Results on the ICO and DR methods are not reported there are no such inference procedures for these methods. Similarly to the simulations for the time in response, the performance of our inference methods was satisfactory in all cases, with the exception of low coverage probabilities for the jackknife estimator when $n=100$. The latter coverage probabilities were at the nominal level when $n=400$. Further simulation results for evaluating the performance of the proposed ITR estimator when $\mathcal{F}$ is the RKHS with the Gaussian kernel are illustrated in Figures~10-12 in Appendix C.2. These results revealed similar patterns to those for the duration of tumor response.


\section{SPECTRUM Trial Data Analysis}
\label{s:analysis}
The proposed methodology was applied to data from the SPECTRUM Trial \citep{Vermorken13}, a phase III randomized trial on recurrent or metastatic squamous-cell carcinoma of the head and neck. The goal of this trial was to evaluate the effectiveness of the addition of panitumumab, a fully human monoclonal antibody which inhibits the epidermal growth factor receptor, to chemotherapy as a first-line treatment approach. The data used in this analysis were obtained from \url{https://www.projectdatasphere.org/} which is maintained by Project Data Sphere and included 520 patients. Of them, 260 patients were randomly assigned to the chemotherapy+panitumumab group ($A=1$) while the remaining patients were assigned to the chemotherapy alone group ($A=-1$). In this trial, tumor response was defined as an at least 30\% decrease in the sum of the longest diameter of target lesions according to the Response Evaluation Criteria in Solid Tumors \citep{RECIST}. Throughout the follow-up period, 138 (26.5\%) patients achieved tumor response while 457 (87.9\%) experienced a progression of their disease and/or died. Among the latter patients, 120 (26.3\%) had achieved tumor response prior to their disease progression or death. The overall median (95\% CI) PFS time was 5.59 (5.29, 5.85) months. The estimates of the treatment-specific cumulative transition intensities and state occupation probabilities are depicted in Figures 13 and 14 in Appendix D. 

In this analysis, the value $\tau$ was set to 18 months (90th percentile of the follow-up times; there were not many transitions to or from the tumor response state after this timepoint) and the preference weight set was $\mathcal{W}=\{w_1,w_2,w_3\}=\{(0,1,0)',(0.5,1,0)',(1,1,0)'\}$. As mentioned in Section~\ref{s:estimation}, in this multistate process setup, the utility under the preference weight $w_1$ corresponds to the restricted mean duration of tumor response, while the choices $w_2$ and $w_3$ provide a restricted quality-adjusted lifetime (where the time spent in the initial state is reduced by $50\%$) and the restricted PFS time, respectively. The covariates considered in this analysis were centered age (in years) at randomization ($Z_{\textrm{age.c}}$), indicator that the primary tumor site is hypopharynx ($Z_{\textrm{hyp}}$), indicator of history of prior treatment for squamous-cell carcinoma of the head and neck ($Z_{\textrm{trt.hist}}$), and indicator that ECOG performance status at baseline indicates symptoms but the patient is ambulatory (vs fully active; $Z_{\textrm{ECOG}}$). The tuning parameters $\lambda_{n,w}$, $w\in\mathcal{W}$, were selected from the candidate set $\{0.001, 0.005, 0.01, 0.1, 0.25, 0.5, 1, 2.5, 5, 10, 20,50, 100\}\times n^{-1/2}$, where $n=520$, using leave-one-out cross validation. All these potential choices for $\lambda_{n,w}$ satisfy the requirements of the theorems in Section~\ref{s:properties}. The estimated optimal decision functions were 
\[
\hat{f}_{n,w_1}(z)=-1.32 -0.54z_{\textrm{age.c}} + 0.90z_{\textrm{hyp}} + 1.17z_{\textrm{trt.hist}} + 0.83z_{\textrm{ECOG}}, 
\]
$\hat{f}_{n,w_2}(z)\approx -1.40 + 2.00z_{\textrm{trt.hist}}$, and $\hat{f}_{n,w_3}(z)\approx -1.47 + 2.00z_{\textrm{trt.hist}}$, where the absolute values of the estimated coefficients of $z_{\textrm{age.c}}$, $z_{\textrm{hyp}}$, and $z_{\textrm{ECOG}}$, in $\hat{f}_{n,w_2}(z)$ and $\hat{f}_{n,w_3}(z)$ were all less than $10^{-6}$. The more complicated form of $\hat{f}_{n,w_1}$ may reflect that there is higher heterogeneity in treatment effect on tumor response compared to the treatment effect on quality-adjusted lifetime and PFS time. Intuitively, one would expect that this higher heterogeneity with respect to tumor response would be carried over to the other outcomes. However, this may not be the case because response has typically a short duration for this tumor and might not have a substantial impact on disease progression and/or death. The class of estimated ITRs is $\{\hat{d}_{n,w}=\textrm{sgn}(\hat{f}_{n,w}):w\in\{w_1,w_2,w_3\}\}$. Clearly, the estimated ITRs $\hat{d}_{n,w_2}$ and $\hat{d}_{n,w_3}$ are equivalent and assign chemotherapy+panitumumab (treatment 1) to the patients with a history of prior treatment, and chemotherapy alone (treatment -1) to those without prior treatment. In contrast, the rule $\hat{d}_{n,w_1}$ is more complicated and accounts for more covariates. Next, we estimated the performances (i.e., value functions) of the estimated optimal ITRs and compared them to those of the fixed, one-size-fits-all, rules $d(z)\coloneqq 1$ (everyone is assigned to chemotherapy+panitumumab) and $d(z)\coloneqq -1$ (everyone is assigned to chemotherapy alone). To account for the multiplicity due to conducting inference about nine parameters in total, we calculated 95\% simultaneous confidence intervals using the approach described in Section~\ref{s:properties}. The percentile $c_{0.05}$ was estimated based on $B=1000$ simulation replications and the corresponding estimate was $\hat{c}_{0.05}=2.59$. The results from this analysis are summarized in Table~\ref{t:analysis}. The estimated restricted mean (95\% CI) potential time in the tumor response state under $\hat{d}_{n,w_1}$ was 2.57 (1.76, 3.38) months. This time was significantly longer than the corresponding time under the fixed rule $d(z)\coloneqq -1$ [difference (95\% CI): 0.83 (0.03, 1.70) months]. The estimated restricted mean (95\% CI) potential quality-adjusted lifetime under $\hat{d}_{n,w_2}$ was 5.00 (4.17, 5.83) months. This time was significantly longer than that under the fixed rule $d(z)\coloneqq -1$ [difference (95\% CI): 0.95 (0.00, 1.90) months]. Finally, the estimated restricted mean (95\% CI) potential PFS time under $\hat{d}_{n,w_3}$ was 7.43 (6.38, 8.49) months. There were no significant differences in the restricted mean PFS time between the optimal rule and $d(z)\coloneqq -1$ [difference (95\% CI): 1.04  (-0.25, 2.33) months] and between $d(z)\coloneqq 1$ and $d(z)\coloneqq -1$ [difference (95\% CI): 0.97 (-0.27, 2.22) months]. Also, no significant differences were observed between the optimal rule and $d(z)\coloneqq 1$. This might be due to a potentially small difference between the effects of chemotherapy alone and chemotherapy+panitumumab among patients for which the two rules assigned different treatment.

\begin{table}
\caption{Analysis of the SPECTRUM Trial: Estimated mean potential utilities (i.e., value functions) under the estimated optimal individualized treatment rules $\hat{d}_{n,w}$ for the preference weights $w\in\{(0,1,0)',(0.5,1,0)',(1,1,0)'\}$, and comparison with those under the fixed, one-size-fits-all, rules $d(z)\coloneqq 1$ (everyone is assigned to chemotherapy+panitumumab) and $d(z)\coloneqq -1$ (everyone is assigned to chemotherapy alone). The corresponding 95\% simultaneous confidence intervals, that adjust for multiplicity due to conducting inference about nine parameters in total, are also provided. ($\mathcal{V}_w(\hat{d}_{n,w})$: mean potential utility under $\hat{d}_{n,w}$ for the preference weight $w$; $\mathcal{V}_w(1)$: mean potential utility for the preference weight $w$ under the fixed rule $d(z)\coloneqq 1$; $\mathcal{V}_w(-1)$: mean potential utility for the preference weight $w$ under the fixed rule $d(z)\coloneqq -1$)}
\label{t:analysis}
\begin{center}
\begin{tabular}{llcc}
\hline
\multicolumn{1}{l}{$w$} &  
\multicolumn{1}{l}{Parameter} &  
\multicolumn{1}{c}{Estimate} &  
\multicolumn{1}{c}{95\% CI}\\ \hline
$(0,1,0)'$&$\mathcal{V}_w(\hat{d}_{n,w})$& 2.570 & (1.763, 3.376) \\
&$\mathcal{V}_w(\hat{d}_{n,w})-\mathcal{V}_w(1)$& 0.109 & (-0.481, 0.700) \\
&$\mathcal{V}_w(\hat{d}_{n,w})-\mathcal{V}_w(-1)$& 0.862 & (0.025, 1.699) \\ [1.5ex]
$(0.5,1,0)'$&$\mathcal{V}_w(\hat{d}_{n,w})$& 4.999 & (4.168, 5.829) \\
&$\mathcal{V}_w(\hat{d}_{n,w})-\mathcal{V}_w(1)$& 0.087 & (-0.539, 0.713) \\
&$\mathcal{V}_w(\hat{d}_{n,w})-\mathcal{V}_w(-1)$& 0.949 & (0.003, 1.896) \\ [1.5ex]
$(1,1,0)'$&$\mathcal{V}_w(\hat{d}_{n,w})$& 7.434 & (6.377, 8.490) \\
&$\mathcal{V}_w(\hat{d}_{n,w})-\mathcal{V}_w(1)$& 0.071 & (-0.739, 0.881) \\
&$\mathcal{V}_w(\hat{d}_{n,w})-\mathcal{V}_w(-1)$& 1.043 & (-0.248, 2.334) \\
\hline
\end{tabular}
\end{center}
\end{table}

\section{Discussion}
\label{s:discuss}

This article addressed the issue of optimal ITR estimation in randomized clinical trials with right-censored multistate processes. To achieve this, we devised a novel objective function that can handle general nonhomogeneous multistate processes and can easily incorporate patient preferences. A key feature of the proposed methodology is that it utilizes information from both uncensored and censored observations without positing and estimating a model for the conditional expectation of the outcome given $A$ and $Z$, as opposed to the methods by \citet{Zhao15} for censored failure times. Optimization of this objective function was based on the outcome weighted learning framework \citep{Zhao12}. 
The simulation studies provided numerical evidence for the validity of the proposed ITR estimation approach and inference procedures. Also, these studies showed a better performance of the proposed method compared to the methods by \citet{Zhao15} for censored failure times.

There are two important practical considerations when applying the proposed approach. First, one needs to specify the preference weight set $\mathcal{W}$. This can be achieved by subject matter experts (e.g., clinicians) or via a survey in a sample of patients. For a given $\mathcal{W}$, the choice of the most appropriate rule within the class of estimated ITRs $\{\hat{d}_{n,w}:w\in\mathcal{W}\}$ for a given patient with a preference weight $w_0$ is $\hat{d}_{n,w}$, with $w=\argmin_{w\in\mathcal{W}}\|w-w_0\|$ (i.e., the closest preference weight in the set $\mathcal{W}$). Second, the class of decision functions $\mathcal{F}$ needs to be chosen. Even though the proposed ITR estimator is universally consistent when $\mathcal{F}$ is the RKHS with the Gaussian kernel, the rate of convergence under this choice may be extremely slow \citep{Steinwart08}. This means that an extremely large sample size may be needed in practice in order to obtain an ITR $\hat{d}_{n,w}$ whose value $\mathcal{V}_w(\hat{d}_{n,w})$ is reasonably close to the optimal value $\mathcal{V}_w(\hat{d}_w^*)$. For this reason, we mainly focused on the class of linear decision functions in most of this article. Under this choice, the proposed estimator is universally consistent only if the true optimal decision function $f_w^*$ is linear. However, even if $f_w^*$ is not linear, the value of the estimated ITR $\mathcal{V}_w(\hat{d}_{n,w})$ can be close to the optimal value. This was illustrated in the simulation studies presented in Section~\ref{s:sims}. In practice, one can further improve the performance of the estimated ITR when $\mathcal{F}$ is the class of linear functions by considering two-way or three-way covariate interactions \citep{Zhou17}. 

A key assumption of the proposed approach is independent censoring. This assumption is realistic in many clinical trials where accrual time is not associated with patient characteristics and censoring is mainly due to administrative reasons \citep{Zhao11,Goldberg17}. A plausible relaxation of this assumption is to allow censoring to depend on treatment $A$, since censoring rate will likely be higher among those receiving the treatment with the greater toxicity \citep{Templeton20}. 
A further relaxation of the independent censoring assumption is to allow censoring to depend on both $A$ and $Z$ by imposing a semiparametric Cox model (more details on both relaxations are provided in Appendix E). 
It must be noted that, for the case of censored failure times, the DR method \citep{Zhao15} allows the censoring model to be misspecified (unlike the proposed approach) provided that the failure time model is correctly specified. However, the true censoring model may be of a less complicated form than the true failure time model in clinical trial applications \citep{Goldberg17} and, thus, more likely to be correctly specified.

An interesting extension of this work is to allow for interval censoring. This could be achieved by utilizing an estimator of the state occupation probabilities with interval-censored data, and deriving an appropriate objective function using, potentially, calculations similar to those in Section 2. 
Also, extending the proposed approach for the single-decision problem to the multi-decision setting, such as a sequential multiple assignment randomized trial (SMART) \citep{Lavori00}, is both practically and methodologically important.



\section*{Acknowledgements}

This article is based on research using data obtained from \url{www.projectdatasphere.org}, which is maintained by Project Data Sphere. Neither Project Data Sphere nor the owner(s) of any information from the web site have contributed to, approved, or are in any way responsible for the contents of this article. This research was supported by grants R01AI140854 and R21AI145662 from the National Institutes of Health and the Indiana University Precision Health Initiative.\vspace*{-8pt}

\bibliographystyle{chicago}
\bibliography{ITR_msm_preprint_v2}

\newpage
\begin{appendices}

\section*{Appendix A: Proofs}
\renewcommand{\thesubsection}{A.\arabic{subsection}}
\label{s:intro}
In this Appendix we provide the proofs of the theorems stated in Section 3 of the main text.  The proposed methodology assumes the following regularity conditions.
\begin{itemize}
\item[C1.] The right censoring time $C$ is independent in the sense that $\{Y_w^*(\cdot;1),Y_w^*(\cdot;-1),A,Z\}\indep C$.
\item[C2.] The benefit process has a square-integrable total variation, i.e. $E\{\int_0^{\tau}|dY_w(t)|\}^2< \infty$.
\item[C3.] The covariate space $\mathcal{Z}$ is a compact subset of $\mathbb{R}^p$.
\item[C4.] The true cumulative baseline hazard $\Lambda_{0}(t)$ of the right censoring distribution is a continuous function on $[0,\tau]$.
\item[C5.] For a linear decision function $\tilde{f}_w(\cdot)=\tilde{\beta}_{0,w}+\langle\tilde{\beta}_{1,w},\cdot\rangle$ that minimizes $\mathcal{R}_{\phi,w}(f)$ over the space of linear functions, $(\tilde{\beta}_{0,w},\tilde{\beta}_{1,w}')'\in \mathcal{B}\subset\mathbb{R}^{p+1}$, where $\mathcal{B}$ is a compact and convex set. Moreover, letting $T_w^*=\int_0^{\tau}Y_{w}(t)I(C\geq T\wedge t)\textrm{d}m(t)$,
\[
P[T_w^*>0,A=1,\{1-(\tilde{\beta}_{0,w}+\langle\tilde{\beta}_{1,w},Z\rangle)\}\{1-(\beta_0+\langle\beta_1,Z\rangle)\}<0]
\] 
\[
+P[T_w^*>0,A=-1,\{1+(\tilde{\beta}_{0,w}+\langle\tilde{\beta}_{1,w},Z\rangle)\}\{1+(\beta_0+\langle\beta_1,Z\rangle)\}<0]>0,
\] 
for any $(\beta_0,\beta_1')'\neq (\tilde{\beta}_{0,w},\tilde{\beta}_{1,w}')'$.
\end{itemize}
Conditions C1, C2, and C4 are standard in the literature of nonparametric methods for survival and multistate process data. A plausible relaxation of condition C1 is to allow censoring to depend on treatment $A$, since censoring will likely be higher among those receiving the treatment with the greater toxicity \citep{Templeton20}. This can be trivially incorporated into the proposed methodology by simply estimating nonparametrically the cumulative hazard of censoring separately for the two treatment groups. A further relaxation of the independent censoring assumption is to allow censoring to depend on both $A$ and $Z$. In this case, one can impose a semiparametric Cox model of the form $\Lambda(t;A,Z)=\Lambda_0(t)\exp\{\theta'(A,Z')'\}$ for the right censoring time, and use the estimated conditional hazard in the proposed objective function and value function estimators. Provided that this model is correctly specified, the theoretical properties of the proposed method still hold, with the exception that $\gamma_i(t)$ in $\psi_{i,w}(f)$ (see Appendix B) is replaced by the influence function of $\sqrt{n}\{\hat{\Lambda}_n(t)\exp(\hat{\theta}_n'(a,z')')-\Lambda(t;a,z)\}$ under partial likelihood estimation. Condition C3 is another common condition that ensures that the covariates are bounded. Condition C5 guarantees the uniqueness of the optimal linear decision function $\tilde{f}_w$. A similar condition has been previously used in the literature of (unweighted) support vector machines \citep{Jiang08}.

For notational simplicity, we omit the subscript $w$, that corresponds to the preference weight, and use the more compact notations $Y(t)$, $Y^*(t;1)$, $Y^*(t;-1)$, $\mathcal{V}(d)$, $\hat{\mathcal{V}}_n(d)$, $\hat{f}_{n}$, $\hat{d}_n$, and $\psi(f)$, for the remainder of Appendix A. The proofs of Theorems 2--4 rely heavily on empirical process theory \citep{Van96, Kosorok08}. For these proofs, we use the standard empirical process theory notation
\[
\mathbb{P}_nf=\frac{1}{n}\sum_{i=1}^nf(D_i),
\]
for any measurable function of the data $f:\mathcal{D}\mapsto\mathbb{R}$, where $\mathcal{D}$ is the sample space, and
\[
Pf=\int_{\mathcal{D}} f\textrm{d}P=E\{f(D)\},
\]
where $P$ is the true probability measure on the Borel $\sigma$-algebra on $\mathcal{D}$. Furthermore, for any function $h$ in the space $D[0,\tau]$ of cadlag functions on $[0,\tau]$, we define the supremum norm $\|h\|_{[0,\tau]}=\sup_{t\in[0,\tau]}|h(t)|$. Also, we define the class of fixed functions 
\[
\mathcal{L}_{\delta}=\{\Lambda:\Lambda\in D[0,\tau] \ \textrm{and is non-decreasing with} \ \Lambda(0)=0, \|\Lambda-\Lambda_0\|_{[0,\tau]}<\delta\},
\]
for some $\delta>0$, and the data-dependent function
\[
\xi_{\Lambda}(D)=\int_0^{\tau}\frac{Y(t)I(C\geq T\wedge t)}{\exp\{-\Lambda(\tilde{T}\wedge t)\}}\textrm{d}m(t).
\]
Before providing the proofs of the theorems listed in the main text, we state and prove a useful lemma. 

\begin{lemma} If conditions C2 and C4 are satisfied, then the class of functions $\{\xi_{\Lambda}:\Lambda\in\mathcal{L}_{\delta}\}$ is $P$-Donsker. 
\end{lemma}
\begin{proof}
The class of functions $\{Y(t):t\in[0,\tau]\}$ is $P$-Donsker as a consequence of condition C2 and Lemma 22.4 in \citet{Kosorok08}. Also, recognizing that 
\[
I(C\geq T\wedge t)=1-I(C< T)I(C< t),
\]
the class $\{I(C\geq T\wedge t):t\in[0,\tau]\}$ is $P$-Donsker by Lemma 4.1 in \citep{Kosorok08}. Next, see that the class $\mathcal{L}_{\delta}$ is uniformly bounded by $\Lambda_0(\tau)+\delta$, where $\Lambda_0(\tau)<\infty$ by condition C4, and consider the data-dependent function $f_{\Lambda}:\mathcal{D}\mapsto[0,\Lambda_0(\tau)+\delta]$ with $f_{\Lambda}(D)=\Lambda(\tilde{T})$, where $\tilde{T}\in[0,\tau]$ and $\Lambda\in\mathcal{L}_{\delta}$. Now, the class $\{f_{\Lambda}:\Lambda\in\mathcal{L}_{\delta}\}$ is $P$-Donsker as a consequence of Lemma 9.11 in \citet{Kosorok08}, as is (trivially) the fixed class $\{\Lambda(t):\Lambda\in\mathcal{L}_{\delta},t\in[0,\tau]\}$. Thus, by Corollary 9.32 in \citet{Kosorok08} and the Lipschitz continuity of the exponential function on $[0,\Lambda_0(\tau)+\delta]$, the class $\{\exp\{f_{\Lambda}\wedge\Lambda(t)\}:\Lambda\in\mathcal{L}_{\delta},t\in[0,\tau]\}$ is also $P$-Donsker. Given that the latter class is uniformly bounded by $\exp\{\Lambda_0(\tau)+\delta\}$ and that $\Lambda(\tilde{T}\wedge t)=\Lambda(\tilde{T})\wedge\Lambda(t)$, it follows that the class $\{\zeta_{\Lambda,t}: \Lambda\in\mathcal{L}_{\delta},t\in[0,\tau]\}$ with
\[
\zeta_{\Lambda,t}(D)=\frac{Y(t)I(C>T\wedge t)}{\exp\{-\Lambda(\tilde{T}\wedge t)\}},
\]
is $P$-Donsker by virtue of being a product of uniformly bounded $P$-Donsker classes. Now, the conclusion of Lemma 1 follows from the fact that 
\[
\left|\int_0^{\tau}\zeta_{\Lambda_1,t}\textrm{d}m(t)-\int_0^{\tau}\zeta_{\Lambda_2,t}\textrm{d}m(t)\right|\leq\sup_{t\in[0,\tau]}\left|\zeta_{\Lambda_1,t}-\zeta_{\Lambda_2,t}\right|\tau,
\]
which implies continuity, and Lemma 15.10 in \citet{Kosorok08}. \hfill $\square$
\end{proof}

\subsection{Proof of Theorem 1 (Fisher consistency)}
Using the same arguments to those used in \citet{Tsiatis19} for the case of a general ITR, it can be easily shown that the optimal ITR in our case satisfies
\begin{eqnarray}
d^*(z)&=&\argmax_{a\in\{-1,1\}}E\left\{\left.\int_0^{\tau}Y^*(t;a)\textrm{d}m(t)\right|Z=z\right\} \nonumber \\
&=&\textrm{sgn}\left[E\left\{\left.\int_0^{\tau}Y^*(t;1)\textrm{d}m(t)\right|Z=z\right\}-E\left\{\left.\int_0^{\tau}Y^*(t;-1)\textrm{d}m(t)\right|Z=z\right\}\right], \label{OITR}
\end{eqnarray}
for all $z\in\mathcal{Z}$. Now, by Tonelli's theorem \citep{Athreya06} and given that
\[
\frac{Y(t)I(C\geq T\wedge t)}{\exp\{-\Lambda_0(T\wedge t)\}}=\frac{Y(t)I(C\geq T\wedge t)}{\exp\{-\Lambda_0(\tilde{T}\wedge t)\}}, \ \ \ \ t\in[0,\tau],
\]
as argued in the main text, the surrogate risk $\mathcal{R}_{\phi}$ can be expressed as
\begin{eqnarray*}
\mathcal{R}_{\phi}(f)&=&\int_0^{\tau}E\left[\frac{Y(t)I(C\geq T\wedge t)\phi(Af(Z))}{\exp\{-\Lambda_0(T\wedge t)\}\{A\pi_0 + (1-A)/2\}}\right]\textrm{d}m(t) \\
&=&\int_0^{\tau}E\left[\frac{Y(t)\phi(Af(Z))E\{I(C\geq T\wedge t)|T,Y(t),A,Z\}}{\exp\{-\Lambda_0(T\wedge t)\}\{A\pi_0 + (1-A)/2\}}\right]\textrm{d}m(t).
\end{eqnarray*}
Given that condition C1 implies that 
\[
E\{I(C\geq T\wedge t)|T,Y(t),A,Z\}=E\{I(C\geq T\wedge t)|T\}=\exp\{-\Lambda_0(T\wedge t)\}, \ \ \ \ t\in[0,\tau],
\]
it follows that
\begin{eqnarray*}
\mathcal{R}_{\phi}(f)&=&\int_0^{\tau}E\left[\frac{Y(t)\phi(Af(Z))}{A\pi_0 + (1-A)/2}\right]\textrm{d}m(t) \\
&=& \int_0^{\tau} E\left\{Y(t)I(A=1)\frac{
\phi(f(Z))}{\pi_0}+Y(t)I(A=-1)\frac{\phi(-f(Z))}{1-\pi_0}\right\}\textrm{d}m(t).
\end{eqnarray*}
Next, by a second application of Tonelli's theorem and assumptions A1--A3 we have
\begin{eqnarray*}
\mathcal{R}_{\phi}(f)&=& \int_0^{\tau} E\left\{Y^*(t;1)I(A=1)\frac{\phi(f(Z))}{\pi_0}\right\}\textrm{d}m(t) \\
&&+\int_0^{\tau}E\left\{Y^*(t;-1)I(A=-1)\frac{\phi(-f(Z))}{1-\pi_0}\right\}\textrm{d}m(t) \\
&=& E\left[\left\{\int_0^{\tau}Y^*(t;1)\textrm{d}m(t)\right\} I(A=1)\frac{\phi(f(Z))}{\pi_0}\right] \\
&&+E\left[\left\{\int_0^{\tau}Y^*(t;-1)\textrm{d}m(t)\right\}I(A=-1)\frac{\phi(-f(Z))}{1-\pi_0}\right] \\
&=&E\left[E\left\{\left.\int_0^{\tau}Y^*(t;1)\textrm{d}m(t)\right|Z\right\}\phi(f(Z))\right. \\
&&+\left.E\left\{\left.\int_0^{\tau}Y^*(t;-1)\textrm{d}m(t)\right|Z\right\}\phi(-f(Z))\right].
\end{eqnarray*}
The function $\tilde{f}$ that minimizes
\[
E\left\{\left.\int_0^{\tau}Y^*(t;1)\textrm{d}m(t)\right|Z\right\}\phi(f(Z))+E\left\{\left.\int_0^{\tau}Y^*(t;-1)\textrm{d}m(t)\right|Z\right\}\phi(-f(Z))
\]
\begin{eqnarray*}
&=&E\left\{\left.\int_0^{\tau}Y^*(t;1)\textrm{d}m(t)\right|Z\right\}\max(0,1-f(Z)) \\ 
&&+E\left\{\left.\int_0^{\tau}Y^*(t;-1)\textrm{d}m(t)\right|Z\right\}\max(0,1+f(Z)) \\
&=&E\left\{\left.\int_0^{\tau}Y^*(t;1)\textrm{d}m(t)\right|Z\right\}\{1-f(Z)\}I(f(Z)\leq -1)  \\
&&+E\left\{\left.\int_0^{\tau}Y^*(t;-1)\textrm{d}m(t)\right|Z\right\}\{1+f(Z)\}I(f(Z)\geq 1)  \\
&&+\left(E\left\{\left.\int_0^{\tau}Y^*(t;1)\textrm{d}m(t)\right|Z\right\}+E\left\{\left.\int_0^{\tau}Y^*(t;-1)\textrm{d}m(t)\right|Z\right\}\right. \\
&&+\left.\left[E\left\{\left.\int_0^{\tau}Y^*(t;-1)\textrm{d}m(t)\right|Z\right\}-E\left\{\left.\int_0^{\tau}Y^*(t;1)\textrm{d}m(t)\right|Z\right\}\right]f(Z)\right) \\
&&\times I(-1<f(Z)<1),
\end{eqnarray*}
also minimizes $\mathcal{R}_{\phi}$. The latter function is a (continuous) piecewise linear function which decreases strictly on $(-\infty,-1]$ and increases strictly on $[1,\infty)$, almost surely. Therefore, the minimizer should lie in $[-1,1]$, almost surely. Consequently, for any $z\in\mathcal{Z}$, if 
\[
E\left\{\left.\int_0^{\tau}Y^*(t;1)dm(t)\right|Z=z\right\}>E\left\{\left.\int_0^{\tau}Y^*(t;-1)dm(t)\right|Z=z\right\},
\]
then $\tilde{f}(z)$ should be positive, and if 
\[
E\left\{\left.\int_0^{\tau}Y^*(t;1)dm(t)\right|Z=z\right\}<E\left\{\left.\int_0^{\tau}Y^*(t;-1)dm(t)\right|Z=z\right\},
\]
then $\tilde{f}(z)$ should be negative. Consequently, by \eqref{OITR}, $d^*(z)=\textrm{sgn}\{\tilde{f}(z)\}$ for all $z\in\mathcal{Z}$.

\subsection{Proof of Theorem 2}
First, define $g_{\pi}(D)=A\pi + (1 - A)/2$ and
\[
L_{f,\Lambda,\pi}(D)=\frac{\xi_{\Lambda}(D)}{g_{\pi}(D)}\phi(Af(Z)),
\]
for any $f\in\mathcal{F}$ (where $\xi_{\Lambda}(D)$ was defined in the beginning of Appendix A), and note that $\mathcal{R}_{\phi}(f)=PL_{f,\Lambda_0,\pi_0}$ and $\hat{\mathcal{R}}_\phi(f)=\mathbb{P}_nL_{f,\hat{\Lambda}_n,\hat{\pi}_n}$. Now, we have that
\begin{eqnarray*}
\frac{Y(t)I(C\geq T\wedge t)}{\exp\{-\hat{\Lambda}_n(\tilde{T}\wedge t)\}}&\leq& \exp\{\hat{\Lambda}_n(\tilde{T}\wedge t)\} \\
&\leq& \exp\{\hat{\Lambda}_n(\tau)\},
\end{eqnarray*}
for any $t\in[0,\tau]$ and all $n\geq 1$. This fact, along with the uniform (outer) almost sure consistency of the Nelson--Aalen estimator (guaranteed by conditions C1 and C4), the facts that 
\[
\int_0^{\tau}\exp\{\hat{\Lambda}_n(\tau)\}\textrm{d}m(t)=\tau\exp\{\hat{\Lambda}_n(\tau)\}<\infty,
\]
for all $n\geq 1$ and 
\[
\int_0^{\tau}\exp\{\Lambda_0(\tau)\}\textrm{d}m(t)=\tau\exp\{\Lambda_0(\tau)\}<\infty,
\]
and the extended dominated converenge theorem \citep{Athreya06}, lead to the conclusion that $\xi_{\hat{\Lambda}_n}(D)\overset{as}{\rightarrow}\xi_{\Lambda_0}(D)$. This result, the almost sure consistency of $\hat{\pi}_n$, and the continuous mapping theorem imply that 
\begin{equation}
\max_{1\leq i\leq n}\left|\frac{\xi_{\hat{\Lambda}_n}(D_i)}{g_{\hat{\pi}_n}(D_i)}-\frac{\xi_{\Lambda_0}(D_i)}{g_{\pi_0}(D_i)}\right|=o_{as}(1) \label{cons1}.
\end{equation}
Consequently,
\begin{equation}
\sup_{f\in\mathcal{F}}\left\{\max_{1\leq i\leq n}\left|L_{f,\hat{\Lambda}_n,\hat{\pi}_n}(D_i)-L_{f,\Lambda_0,\pi_0}(D_i)\right|\right\}=o_{as}(1), \label{cons2}
\end{equation}
Next, define 
\[
\tilde{f}\in \argmin_{f\in\mathcal{F}}\mathcal{R}_{\phi}(f).
\]
By the positivity of $\lambda_n$ and the definition of $\hat{f}_n$, we have that
\begin{eqnarray*}
\mathbb{P}_nL_{\hat{f}_n,\hat{\Lambda}_n,\hat{\pi}_n}&\leq& \mathbb{P}_nL_{\hat{f}_n,\hat{\Lambda}_n,\hat{\pi}_n}+\lambda_n\|\hat{f}_n\|^2 \\
&\leq& \mathbb{P}_nL_{\tilde{f},\hat{\Lambda}_n,\hat{\pi}_n}+\lambda_n\|\tilde{f}\|^2 \\
&=&\mathbb{P}_nL_{\tilde{f},\Lambda_0,\pi_0}+\lambda_n\|\tilde{f}\|^2+o_{as}(1)
\end{eqnarray*}
Taking limit superiors in both sides we get
\[
\limsup_{n\rightarrow\infty}\mathbb{P}_nL_{\hat{f}_n,\hat{\Lambda}_n,\hat{\pi}_n}\leq PL_{\tilde{f},\Lambda_0,\pi_0}=\mathcal{R}_{\phi}(\tilde{f}),
\]
almost surely, by the strong law of large numbers and the fact that $\lambda_n\rightarrow 0$. This implies that, fall all $n$ sufficiently large,
\[
\mathbb{P}_nL_{\hat{f}_n,\hat{\Lambda}_n,\hat{\pi}_n}\leq \mathcal{R}_{\phi}(\tilde{f})\leq PL_{\hat{f}_n,\Lambda_0,\pi_0}=\mathcal{R}_{\phi}(\hat{f}_n),
\]
almost surely, by the definition of $\tilde{f}$. Therefore, for all $n$ sufficiently large, we have 
\begin{eqnarray}
|\mathcal{R}_{\phi}(\hat{f}_n)-\mathcal{R}_{\phi}(\tilde{f})|&\leq& \left|\mathbb{P}_nL_{\hat{f}_n,\hat{\Lambda}_n,\hat{\pi}_n}-PL_{\hat{f}_n,\Lambda_0,\pi_0}\right| \nonumber \\
&\leq& \left|(\mathbb{P}_n-P)L_{\hat{f}_n,\Lambda_0,\pi_0}\right|+\mathbb{P}_n\left|L_{\hat{f}_n,\hat{\Lambda}_n,\hat{\pi}_n}-L_{\hat{f}_n,\Lambda_0,\pi_0}\right| \nonumber \\
&\leq& \left|(\mathbb{P}_n-P)L_{\hat{f}_n,\Lambda_0,\pi_0}\right|+\delta, \label{ineq1}
\end{eqnarray}
almost surely, for any $\delta>0$ by \eqref{cons2}. Next, it needs to be shown that $(\mathbb{P}_n-P)L_{\hat{f}_n,\Lambda_0,\pi_0}=o_p(1)$. By the definition of $\hat{f}_n$ we have,
\[
\mathbb{P}_nL_{\hat{f}_n,\hat{\Lambda}_n,\hat{\pi}_n}+\lambda_n\|\hat{f}_n\|^2\leq \mathbb{P}_nL_{f,\hat{\Lambda}_n,\hat{\pi}_n}+\lambda_n\|f\|^2,
\]
for any $f\in\mathcal{F}$. Selecting $f\equiv 0$, we have
\[
\mathbb{P}_nL_{\hat{f}_n,\hat{\Lambda}_n,\hat{\pi}_n}+\lambda_n\|\hat{f}_n\|^2\leq \mathbb{P}_n\frac{\xi_{\hat{\Lambda}_n}}{g_{\hat{\pi}_n}},
\]
since $\phi(0)=1$ and $\|0\|=0$. By the non-negativity of $L_{\hat{f}_n,\hat{\Lambda}_n,\hat{\pi}_n}(D)$, the previous inequality implies that
\begin{eqnarray*}
\lambda_n\|\hat{f}_n\|^2&\leq& \mathbb{P}_n\frac{\xi_{\hat{\Lambda}_n}}{g_{\hat{\pi}_n}} \\
&=&\mathbb{P}_n\frac{\xi_{\Lambda_0}}{g_{\pi_0}}+o_{as}(1) \\
&\leq& \frac{1}{n}\sum_{i=1}^n\frac{\tau\exp\{\Lambda_0(\tau)\}}{A_i\pi_0+(1-A_i)/2}+o_{as}(1) \\
&\leq&\frac{\tau\exp\{\Lambda_0(\tau)\}}{c_1\wedge(1-c_2)}+o_{as}(1),
\end{eqnarray*}
by assumption A3. Therefore, for all sufficiently large $n$ and any $\delta'>0$, we have 
\[
\lambda_n\|\hat{f}_n\|^2\leq \frac{\tau\exp\{\Lambda_0(\tau)\}}{c_1\wedge(1-c_2)}+\delta'\equiv M_{\delta'},
\]
almost surely. Now, the class of functions 
\[
\mathcal{G}_1(\delta')=\{\sqrt{\lambda_n}f:f\in\mathcal{F},\|\sqrt{\lambda_n}f\|\leq \sqrt{M_{\delta'}}\}
\]
is Donsker, for any $\delta'>0$ small. This follows from the fact that if $\mathcal{F}$ is the class of linear or polynomial functions, then $\mathcal{F}$ is Donsker by Lemma 9.6 and Theorem 9.2 in \citet{Kosorok08}, and subsets of Donsker class are also Donsker. If $\mathcal{F}$ is a RKHS, then the $\mathcal{G}_1(\delta')$ is Donsker by condition C3 and similar arguments to those used in the proof of Lemma A.9 in \citet{Hable12}. Also, the class 
\[
\mathcal{G}_2(\delta')=\{\sqrt{\lambda_n}L_{f,\Lambda_0,\pi_0}:f\in\mathcal{F},\|\sqrt{\lambda_n}f\|\leq \sqrt{M_{\delta'}}\}
\]
is Donsker by Corollary 9.32 in \citet{Kosorok08}, since $L_{f,\Lambda_0,\pi_0}$ is Lipschitz continuous in $f$, since the hinge loss is Lipschitz continuous in $f$. This implies that, for all $n$ sufficiently large, $\sqrt{\lambda_n}L_{\hat{f}_n;\Lambda_0,\pi_0}$ belongs to the Donsker class $\mathcal{G}_2(\delta')$ and therefore 
\[
\sqrt{n}(\mathbb{P}_n-P)\sqrt{\lambda_n}L_{\hat{f}_n,\Lambda_0,\pi_0}=O_p(1).
\]
Consequently,
\[
(\mathbb{P}_n-P)L_{\phi}(\hat{f}_n;\Lambda_0,\pi_0)=\frac{1}{\sqrt{n\lambda_n}}O_p(1)=o_p(1),
\]
since $n\lambda_n\rightarrow\infty$ by assumption. Substituting this result in \eqref{ineq1} we have
\[
|\mathcal{R}_{\phi}(\hat{f}_n)-\mathcal{R}_{\phi}(\tilde{f})|\leq o_p(1)+\delta.
\]
Setting $\delta=\delta_n\downarrow 0$, with $\sqrt{n}\delta_n\rightarrow\infty$, implies that 
\[
|\mathcal{R}_{\phi}(\hat{f}_n)-\mathcal{R}_{\phi}(\tilde{f})|=o_p(1),
\]
which completes the proof of the first statement of Theorem 2. For the second statement, 
it can be shown using similar arguments to those used in the proof of Theorem 3.2 in \citet{Zhao12} that
\[
\mathcal{R}(f)-\mathcal{R}(f^*)\leq \mathcal{R}_{\phi}(f)-\mathcal{R}_{\phi}(f^*),
\]
for any distribution $P$ of the data $D$ and any measurable decision function $f:\mathcal{Z}\mapsto\mathbb{R}$. Therefore, 
\begin{eqnarray}
|\mathcal{V}(\hat{d}_n)-\mathcal{V}(d^*)|&=&|\mathcal{R}(\hat{f}_n)-\mathcal{R}(f^*)| \nonumber \\
&\leq& |\mathcal{R}_{\phi}(\hat{f}_n)-\mathcal{R}_{\phi}(f^*)|. \label{ineq2}
\end{eqnarray}
Now, if $\mathcal{F}$ is the space of linear functions and $f^*\in\mathcal{F}$, then $\inf_{f\in\mathcal{F}}\mathcal{R}_{\phi}(f)=\mathcal{R}_{\phi}(\tilde{f})=\mathcal{R}_{\phi}(f^*)$. Thus, consistency follows from the first statement of Theorem 2 and \eqref{ineq2}. Next, suppose that $\mathcal{F}$ is the RKHS with the Gaussian kernel and that the marginal distribution $\mu$ of $Z$ is regular. Since the Gaussian kernel is a universal kernel, using similar arguments to those used in the proof of Lemma 3.4 in \citet{Zhou17} leads to 
\[
\inf_{f\in\mathcal{F}}\mathcal{R}_{\phi}(f)=\mathcal{R}_{\phi}(f^*).
\]
This, \eqref{ineq2}, and the first statement of Theorem 2 imply universal consistency.

\subsection{Proof of Theorem 3}
Here we provide the proof of the stronger (uniform) statement of Theorem 3 for the case where $\mathcal{F}$ is the space of linear functions. The proof of the weaker (pointwise) statement for any given measurable function $f$ is a simplified version of the proof below and is not discussed further.

First, note that, for any $f\in\mathcal{F}$, we have
\[
\mathcal{V}(\textrm{sgn}(f))=P\frac{\xi_{\Lambda_0}}{g_{\pi_0}}u_f,
\]
where 
\[
u_f(D)=I\{A=\textrm{sgn}(f(Z))\}=I\{Af(Z)\geq 0\},
\]
and 
\[
\hat{\mathcal{V}}_n(\textrm{sgn}(f))=\mathbb{P}_n\frac{\xi_{\hat{\Lambda}_n}}{g_{\hat{\pi}_n}}u_f.
\]
Now, straightforward algebra leads to
\begin{eqnarray}
\sqrt{n}\left\{\hat{\mathcal{V}}_n(\textrm{sgn}(f))-\mathcal{V}(\textrm{sgn}(f))\right\} &=& \sqrt{n}\mathbb{P}_n\left(\frac{1}{g_{\hat{\pi}_n}}-\frac{1}{g_{\pi_0}}\right)\left(\xi_{\hat{\Lambda}_n}-\xi_{\Lambda_0}\right)u_f \nonumber\\
&& +\sqrt{n}\mathbb{P}_n\frac{\xi_{\hat{\Lambda}_n}-\xi_{\Lambda_0}}{g_{\pi_0}}u_f \nonumber\\
&& +\sqrt{n}\mathbb{P}_n\xi_{\Lambda_0}\left(\frac{1}{g_{\hat{\pi}_n}}-\frac{1}{g_{\pi_0}}\right)u_f \nonumber\\
&&+\sqrt{n}(\mathbb{P}_n-P)\frac{\xi_{\Lambda_0}}{g_{\pi_0}}u_f \nonumber \\
&\equiv& B_{n,1}(f)+B_{n,2}(f)+B_{n,3}(f)+B_{n,4}(f) \label{decomp}
\end{eqnarray}
Next, for any functional $h:\mathcal{F}\mapsto\mathbb{R}$, define the supremum norm $\|h\|_{\mathcal{F}}=\sup_{f\in\mathcal{F}}|h(f)|$. For the term $B_{n,1}(f)$ in \eqref{decomp} we have
\begin{eqnarray*}
\|B_{n,1}\|_{\mathcal{F}}&\leq&\left|\sqrt{n}\mathbb{P}_n\left(\frac{1}{g_{\hat{\pi}_n}}-\frac{1}{g_{\pi_0}}\right)\left(\xi_{\hat{\Lambda}_n}-\xi_{\Lambda_0}\right)\right| \\
&\leq&\sqrt{n}\max_{1\leq i\leq n}\left|\frac{1}{g_{\hat{\pi}_n}(D_i)}-\frac{1}{g_{\pi_0}(D_i)}\right|\sup_{t\in[0,\tau]}\left|\exp\{\hat{\Lambda}_n(t)\}-\exp\{\Lambda_0(t)\}\right|\tau  \\
&\leq& \max\left\{\frac{1}{\hat{\pi}_n\pi_0},\frac{1}{(1-\hat{\pi}_n)(1-\pi_0)}\right\}\left|\sqrt{n}(\hat{\pi}_n-\pi_0)\right| \\
&&\times\sup_{t\in[0,\tau]}\left|\exp\{\hat{\Lambda}_n(t)\}-\exp\{\Lambda_0(t)\}\right|\tau.
\end{eqnarray*}
The last inequality, along with the boundedness of $\max[(\hat{\pi}_n\pi_0)^{-1},\{(1-\hat{\pi}_n)(1-\pi_0)\}^{-1}]$ for all $n$ sufficiently large, as a result of assumption A3, the fact that $\sqrt{n}(\hat{\pi}_n-\pi_0)=O_p(1)$, as a consequence of the central limit theorem, the fact that 
\[
\sup_{t\in[0,\tau]}\left|\exp\{\hat{\Lambda}_n(t)\}-\exp\{\Lambda_0(t)\}\right|=o_{as*}(1),
\]
as a result of the strong uniform consistency of the Nelson--Aalen estimator of the cumulative hazard (guaranteed by conditions C1 and C4) and the continuous mapping theorem, and the boundedness of the length of the follow-up interval $\tau$, lead to the conclusion that $\|B_{n,1}\|_{\mathcal{F}}=o_p(1)$. 

The term $B_{n,2}(f)$ can be expressed as follows
\begin{equation}
B_{n,2}(f)=\sqrt{n}(\mathbb{P}_n-P)\frac{\xi_{\hat{\Lambda}_n}-\xi_{\Lambda_0}}{g_{\pi_0}}u_f + \sqrt{n}P\frac{\xi_{\hat{\Lambda}_n}-\xi_{\Lambda_0}}{g_{\pi_0}}u_f. \label{B2_f}
\end{equation}
The class of functions 
\[
\mathcal{F}=\{f(\cdot)=\beta_0+\langle{\beta,\cdot}\rangle:\beta_0\in\mathbb{R},\beta\in\mathbb{R}^p\}
\]
is a Vapnik-{\v C}ervonenkis (VC) class by Lemma 9.6 in \citet{Kosorok08}. This along with lemma 9.9 and theorem 9.2 in \citet{Kosorok08} imply that the class $\{u_f:f\in\mathcal{F}\}$ has bounded uniform entropy integral. In addition, the latter class can be easily argued to be pointwise measurable and, therefore, this class is $P$-Donsker. By conditions C2 and C4 and Lemma 1, the fact that the class $\{\xi_{\Lambda}:\Lambda\in\mathcal{L}_{\delta}\}$ is uniformly bounded by $\tau\exp\{\Lambda_0(\tau)+\delta\}$, assumption A3, and the fact that products of uniformly bounded $P$-Donsker classes are also $P$-Donsker, it follows that the class
\[
\left\{\frac{\xi_{\Lambda}-\xi_{\Lambda_0}}{g_{\pi_0}}u_f:\Lambda\in\mathcal{L}_{\delta},f\in\mathcal{F}\right\}
\]
is $P$-Donsker for some $\delta>0$. Also, by assumption A3, we have 
\begin{eqnarray*}
\sup_{f\in\mathcal{F}}P\left(\frac{\xi_{\Lambda}-\xi_{\Lambda_0}}{g_{\pi_0}}u_f\right)^2&\leq&\left\{\frac{\tau\|\exp(\Lambda)-\exp(\Lambda_0)\|_{[0,\tau]}}{c_1\wedge(1-c_2)}\right\}^2 \\
&\leq&\left[\frac{\tau\exp\{\Lambda_0(\tau)\}\{\exp(\|\Lambda-\Lambda_0\|_{[0,\tau]})-1\}}{c_1\wedge(1-c_2)}\right]^2\rightarrow 0
\end{eqnarray*}
as $\|\Lambda-\Lambda_0\|_{[0,\tau]}\rightarrow 0$. The last two results along with the uniform consistency of the Nelson--Aalen estimator, guaranteed by conditions C1 and C4, and arguments similar to those used in the proof of Lemma 3.3.5 in \citet{Van96} lead to the conclusion that 
\[
\left\|\sqrt{n}(\mathbb{P}_n-P)\frac{\xi_{\hat{\Lambda}_n}-\xi_{\Lambda_0}}{g_{\pi_0}}u_f\right\|_{\mathcal{F}}=o_p(1).
\]
For the second term in the right side of \eqref{B2_f}, it can be shown that the map $\Lambda\mapsto Pg_{\pi_0}^{-1}\xi_{\Lambda}u_f$ is Hadamard differentiable at $\Lambda_0$ with derivative
\[
\eta_{\Lambda_0,f}'(h)=Pg_{\pi_0}^{-1}u_f\int_0^{\tau}Y(t)I(C\geq T\wedge t)\exp\{\Lambda_0(\tilde{T}\wedge t)\}h(\tilde{T}\wedge t)\textrm{d}m(t).
\]
This, along with the functional delta method \citep[Theorem 3.9.4 in][]{Van96} and the fact that 
\[
\sqrt{n}\{\hat{\Lambda}_n(t)-\Lambda_0(t)\}=\frac{1}{\sqrt{n}}\sum_{i=1}^n\gamma_i(t)+o_p(1), \ \ \ \ t\in[0,\tau],
\]
where
\[
\gamma_i(t)=\int_0^t\frac{\textrm{d}N_i(s)}{PY(s)}-\int_0^t\frac{Y_i(s)}{PY(s)}\textrm{d}\Lambda_0(s), 
\]
lead to the conclusion that 
\begin{eqnarray*}
\sqrt{n}P\frac{\xi_{\hat{\Lambda}_n}-\xi_{\Lambda_0}}{g_{\pi_0}}u_f&=&\frac{1}{\sqrt{n}}\sum_{i=1}^n\eta_{\Lambda_0,f}'(\gamma_i)+o_p(1) \\
&=&\frac{1}{\sqrt{n}}\sum_{i=1}^nPg_{\pi_0}^{-1}u_f\int_0^{\tau}Y(t)I(C\geq T\wedge t)\exp\{\Lambda_0(\tilde{T}\wedge t)\}\gamma_i(\tilde{T}\wedge t)dm(t) \\
&&+o_p(1),
\end{eqnarray*}
where $\gamma_i$ is considered fixed under the expectation operator $P$. Therefore,
\[
B_{n,2}(f)=\frac{1}{\sqrt{n}}\sum_{i=1}^n\eta_{\Lambda_0,f}'(\gamma_i)+o_p(1)+\epsilon_{n,1}(f)
\]
where 
\[
\epsilon_{n,1}(f)=\sqrt{n}(\mathbb{P}_n-P)\frac{\xi_{\hat{\Lambda}_n}-\xi_{\Lambda_0}}{g_{\pi_0}}u_f
\]
with $\|\epsilon_{n,1}\|_{\mathcal{F}}=o_p(1)$.

For the term $B_{n,3}(f)$, we have 
\begin{eqnarray}
B_{n,3}(f)&=&\sqrt{n}\mathbb{P}_n\xi_{\Lambda_0}\left(\frac{1}{g_{\hat{\pi}_n}}-\frac{1}{g_{\pi_0}}\right)u_f \nonumber \\
&=&-\left(\mathbb{P}_n\xi_{\Lambda_0}\frac{A}{g_{\pi_n^*}^2}u_f\right) \sqrt{n}(\hat{\pi}_n-\pi_0) \nonumber \\
&=& -\left(\mathbb{P}_n\xi_{\Lambda_0}\frac{A}{g_{\pi_n^*}^2}u_f-P\xi_{\Lambda_0}\frac{A}{g_{\pi_0}^2}u_f\right) \sqrt{n}(\hat{\pi}_n-\pi_0) \nonumber\\
&&-\left(P\xi_{\Lambda_0}\frac{A}{g_{\pi_0}^2}u_f\right)\sqrt{n}(\hat{\pi}_n-\pi_0) , \label{B3_f}
\end{eqnarray}
where $|\pi_n^*-\pi_0|\leq|\hat{\pi}_n-\pi_0|$. For the first term in the right side of \eqref{B3_f} we have
\begin{eqnarray*}
\sup_{f\in\mathcal{F}}\left|\mathbb{P}_n\xi_{\Lambda_0}\frac{A}{g_{\pi_n^*}^2}u_f-P\xi_{\Lambda_0}\frac{A}{g_{\pi_0}^2}u_f\right|&\leq&\sup_{f\in\mathcal{F}}\left|\mathbb{P}_n\xi_{\Lambda_0}u_FA\left(\frac{1}{g_{\pi_n^*}^2}-\frac{1}{g_{\pi_0}^2}\right)\right| \\
&&+\sup_{f\in\mathcal{F}}\left|(\mathbb{P}_n-P)\xi_{\Lambda_0}\frac{A}{g_{\pi_0}^2}u_f\right| \\
&\leq&\max\left\{\frac{1}{\pi_n^{*2}}-\frac{1}{\pi_0^2},\frac{1}{(1-\pi_n^*)^2}-\frac{1}{(1-\pi_0)^2}\right\}\tau\exp\{\Lambda_0(\tau)\} \\
&&+\sup_{f\in\mathcal{F}}\left|(\mathbb{P}_n-P)\xi_{\Lambda_0}\frac{A}{g_{\pi_0}^2}u_f\right| \\
\end{eqnarray*}
The first term in the right side of the above inequality is $o_{as}(1)$ by assumption A3, the fact that $|\pi_n^*-\pi_0|\leq|\hat{\pi}_n-\pi_0|$, the strong law of large numbers, the continuous mapping theorem, and condition C4, which guarantees the finiteness of $\Lambda_0(\tau)$. The second term in the right side of the last inequaltity is $o_{as*}(1)$ as a consequence of the $P$-Donsker property of the class 
\[
\left\{\xi_{\Lambda_0}\frac{A}{g_{\pi_0}^2}u_f:f\in\mathcal{F}\right\},
\]
which follows from similar arguments to those used in the analysis of the term $B_{n,2}(f)$ above, since this property implies that the latter class is also $P$-Glivenko--Cantelli. Therefore, using the last inequality gives
\[
\sup_{f\in\mathcal{F}}\left|\mathbb{P}_n\xi_{\Lambda_0}\frac{A}{g_{\pi_n^*}^2}u_f-P\xi_{\Lambda_0}\frac{A}{g_{\pi_0}^2}u_f\right|=o_{as*}(1).
\]
Consequently, given that $\sqrt{n}(\hat{\pi}_n-\pi_0)=O_p(1)$ by the central limit theorem, it follows from \eqref{B3_f} that
\[
B_{n,3}(f)=-\left(P\xi_{\Lambda_0}\frac{A}{g_{\pi_0}^2}u_f\right)\sqrt{n}(\hat{\pi}_n-\pi_0)+\epsilon_{n,2}(f),
\]
where 
\[
\epsilon_{n,2}(f)=-\left(\mathbb{P}_n\xi_{\Lambda_0}\frac{A}{g_{\pi_n^*}^2}u_f-P\xi_{\Lambda_0}\frac{A}{g_{\pi_0}^2}u_f\right) \sqrt{n}(\hat{\pi}_n-\pi_0),
\]
with $\|\epsilon_{n,2}\|_{\mathcal{F}}=o_{as*}(1)O_p(1)=o_p(1)$.

Taking all the pieces together, it follows from \eqref{decomp} that
\begin{eqnarray*}
\sqrt{n}\left\{\hat{\mathcal{V}}_n(\textrm{sgn}(f))-\mathcal{V}(\textrm{sgn}(f))\right\} &=&\frac{1}{\sqrt{n}}\sum_{i=1}^n\Bigg[\eta_{\Lambda_0,f}'(\gamma_i)-\left(P\xi_{\Lambda_0}\frac{A}{g_{\pi_0}^2}u_f\right)\{I(A_i=1)-\pi_0\} \\
&&+\left\{\frac{\xi_{\Lambda_0}(D_i)}{g_{\pi_0}(D_i)}u_f(D_i)-\mathcal{V}(\textrm{sgn}(f))\right\}\Bigg]+\epsilon_{n}(f) \\
&=&\frac{1}{\sqrt{n}}\sum_{i=1}^n\psi_i(f)+\epsilon_{n}(f) ,
\end{eqnarray*}
where $\epsilon_{n}(f)=B_{n,1}(f)+\epsilon_{n,1}(f)+\epsilon_{n,2}(f)+o_p(1)$ with $\|\epsilon_{n}\|_{\mathcal{F}}=o_p(1)$. Finally, the class of functions $\{\psi(f):f\in\mathcal{F}\}$ is $P$-Donsker as a consequence of the $P$-Donsker property of the class $\{u_f:f\in\mathcal{F}\}$ as argued above, Lemma 15.10 in \citet{Kosorok08}, and the fact that sums of $P$-Donsker classes which are multiplied by random variables with finite second moments are also $P$-Donsker.

\subsection{Proof of Theorem 4}
Let $\hat{f}_n(\cdot)=\hat{\beta}_{n,0}+\langle\hat{\beta}_{n,1},\cdot\rangle$, $\hat{\beta}_n=(\hat{\beta}_{n,0},\hat{\beta}_{n,1}')'$ and $\tilde{\beta}=(\tilde{\beta}_0,\tilde{\beta}_1')'$. Then, by Theorem 2, condition C5, and similar arguments to those used in \citet{Jiang08}, it follows that $\tilde{\beta}$ is unique and $\|\hat{\beta}_n-\tilde{\beta}\|_2=o_p(1)$, where $\|\cdot\|_2$ is the Euclidean norm. Next, using the notation $f_{\beta}(\cdot)=\beta_0+\langle\beta_1',\cdot\rangle$, $\beta=(\beta_0,\beta_1')'$, Thorem 3 guarantees that the class 
\[
\left\{\psi(f_{\beta})-\psi(f_{\tilde{\beta}}):\|\beta-\tilde{\beta}\|_2<\delta\right\}
\]
is $P$-Donsker for any $\delta>0$. Moreover, by the assumption that $P(f_{\tilde{\beta}}(Z)=0)=0$ which implies the continuity of the map $\beta\mapsto I(af_{\beta}(z)\geq 0)$ at $\tilde{\beta}$ for almost all $z\in\mathcal{Z}$, it follows that
\[
P\{\psi(f_{\beta})-\psi(f_{\tilde{\beta}})\}^2\rightarrow 0 \ \ \textrm{as} \ \ \beta\rightarrow\tilde{\beta}.
\]
Next, by Theorem 3 and arguments similar to those used in the proof of Lemma 3.3.5 in \citet{Van96} it follows that
\[
\left|\sqrt{n}\left\{\hat{\mathcal{V}}_n(\textrm{sgn}(f_{\hat{\beta}_n}))-\mathcal{V}(\textrm{sgn}(f_{\hat{\beta}_n}))\right\}-\sqrt{n}\left\{\hat{\mathcal{V}}_n(\textrm{sgn}(f_{\tilde{\beta}}))-\mathcal{V}(\textrm{sgn}(f_{\tilde{\beta}}))\right\}\right|=o_p(1),
\]
for any preference weight $w$ that satisfies the requirements of Section 2 in the main text. Finally, the conclusion of Theorem 4 follows from the last result and the fact that the set $\mathcal{W}$ of preference weights is finite.

\section*{Appendix B: Empirical Versions of the Influence Functions}
\renewcommand{\thesubsection}{B.\arabic{subsection}}

The influence functions $\psi_i(f)$ have the form
\begin{eqnarray*}
\psi_{i,w}(f)&=&\int_0^{\tau}\frac{Y_{i,w}(t)I(C_i\geq T_i\wedge t)I[A_i=\textrm{sgn}\{f(Z_i)\}]}{\exp\{-\Lambda_0(\tilde{T}_i\wedge t)\}\{A_i\pi_0 + (1-A_i)/2\}}\textrm{d}m(t)-\mathcal{V}_w(\textrm{sgn}(f)) \\
&& -E\left\{A\int_0^{\tau}\frac{Y_{w}(t)I(C\geq T\wedge t)I[A=\textrm{sgn}\{f(Z)\}]}{\exp\{-\Lambda_0(\tilde{T}\wedge t)\}\{A\pi_0 + (1-A)/2\}^2}\textrm{d}m(t)\right\}\{I(A_i=1)-\pi_0\} \\
&&+\eta_{\Lambda_0,f,w}'(\gamma_i), \ \ \ \ i=1,\ldots,n, \ \ w\in\mathcal{W}, \ \ f\in\mathcal{F},
\end{eqnarray*}
where
\[
\eta_{\Lambda_0,f,w}'(h)=E\left\{\int_0^{\tau}\frac{Y_{w}(t)I(C\geq T\wedge t)I[A=\textrm{sgn}\{f(Z)\}]}{\exp\{-\Lambda_0(\tilde{T}\wedge t)\}\{A\pi_0 + (1-A)/2\}}h(\tilde{T}\wedge t)\textrm{d}m(t)\right\},
\]
for $h$ in the space $D[0,\tau]$ of right continuous functions on $[0,\tau]$ with left hand limits and
\[
\gamma_i(t)=\int_0^t\frac{dN_i(s)}{EY(s)}-\int_0^t\frac{Y_i(s)}{EY(s)}\textrm{d}\Lambda_0(s), \ \ \ \ i=1,\ldots,n, \ \ t\in[0,\tau].
\]
The empirical versions of the influence functions are
\begin{eqnarray*}
\hat{\psi}_{i,w}(f)&=&\int_0^{\tau}\frac{Y_{i,w}(t)I(C_i>T_i\wedge t)I[A_i=\textrm{sgn}\{f(Z_i)\}]}{\exp\{-\hat{\Lambda}_n(\tilde{T}_i\wedge t)\}\{A_i\hat{\pi}_n + (1-A_i)/2\}}\textrm{d}m(t)-\hat{\mathcal{V}}_{n,w}(\textrm{sgn}(f)) \\
&&-\left\{\frac{1}{n}\sum_{j=1}^nA_j\int_0^{\tau}\frac{Y_{j,w}(t)I(C_j>T_j\wedge t)I[A_j=\textrm{sgn}\{f(Z_j)\}]}{\exp\{-\hat{\Lambda}_n(\tilde{T}_j\wedge t)\}\{A_j\hat{\pi}_n + (1-A_j)/2\}^2}\textrm{d}m(t)\right\}\{I(A_i=1)-\hat{\pi}_n\} \\
&&+\hat{\eta}_{\hat{\Lambda}_n,f,w}'(\hat{\gamma}_i), \ \ \ \ i=1,\ldots,n, \ \ w\in\mathcal{W}, \ \ f\in\mathcal{F},
\end{eqnarray*}
where
\[
\hat{\eta}_{\hat{\Lambda}_n,f,w}'(h)=\frac{1}{n}\sum_{i=1}^n\left\{\int_0^{\tau}\frac{Y_{i,w}(t)I(C_i>T_i\wedge t)I[A_i=\textrm{sgn}\{f(Z_i)\}]}{\exp\{-\hat{\Lambda}_n(\tilde{T}_i\wedge t)\}\{A_i\hat{\pi}_n + (1-A_i)/2\}}h(\tilde{T}_i\wedge t)\textrm{d}m(t)\right\},
\]
for $h\in D[0,\tau]$, and
\[
\hat{\gamma}_i(t)=\int_0^t\frac{\textrm{d}N_i(s)}{n^{-1}\sum_{j=1}^nY_j(s)}-\int_0^t\frac{Y_i(s)}{n^{-1}\sum_{j=1}^nY_j(s)}\textrm{d}\hat{\Lambda}_n(s), \ \ \ \ i=1,\ldots,n, \ \ t\in[0,\tau].
\]

\section*{Appendix C. Additional Simulation Results}
\renewcommand{\thesubsection}{C.\arabic{subsection}}
\setcounter{subsection}{0}

In this Appendix we provide additional simulation results. 

\subsection{Additional simulation results under $w=(0,1,0)'$ (duration of tumor response)}

Simulation results evaluating the effect of selecting different values of $\tau$ in the analysis are illustrated in Figures~\ref{f:sim_taus_low}--\ref{f:sim_taus_high}.

\begin{figure}
\centerline{\includegraphics[width=6in]{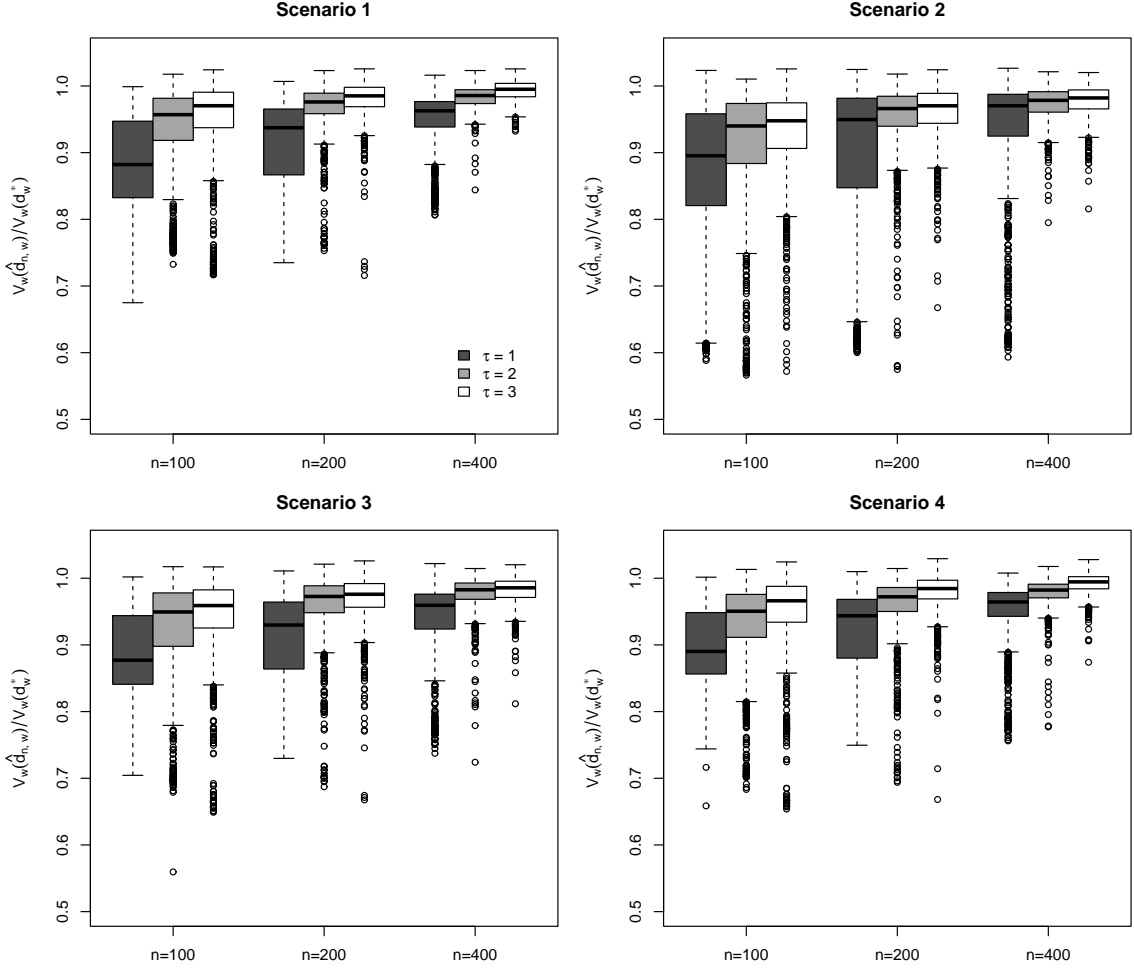}}
\caption{Simulation study: Performance of $\hat{d}_{n,w}$ for the duration of tumor response (i.e., $w=(0,1,0)'$), in terms of the estimated individualized treatment rule value ratio $\mathcal{V}_w(\hat{d}_{n,w})/\mathcal{V}_w(d_w^*)$ for different choices of the $\tau$ used in the analysis. The total length of the follow-up is equal to 3. Results under an average censoring rate of 28.4\%.}
\label{f:sim_taus_low}
\end{figure}

\begin{figure}
\centerline{\includegraphics[width=6in]{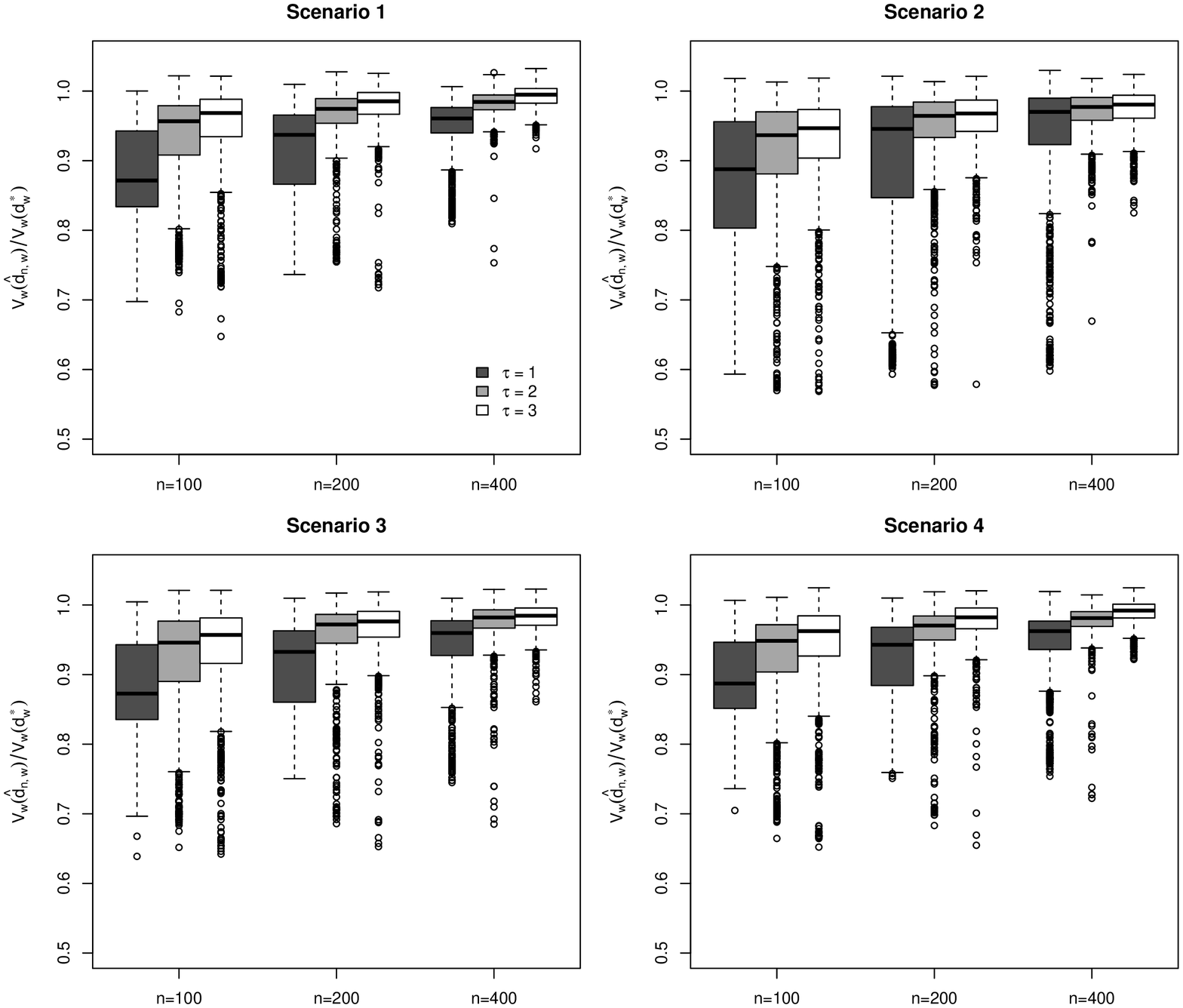}}
\caption{Simulation study: Performance of $\hat{d}_{n,w}$ for the duration of tumor response (i.e., $w=(0,1,0)'$), in terms of the estimated individualized treatment rule value ratio $\mathcal{V}_w(\hat{d}_{n,w})/\mathcal{V}_w(d_w^*)$ for different choices of the $\tau$ used in the analysis. The total length of the follow-up is equal to 3. Results under an average censoring rate of 42.8\%.}
\label{f:sim_taus_med}
\end{figure}

\begin{figure}
\centerline{\includegraphics[width=6in]{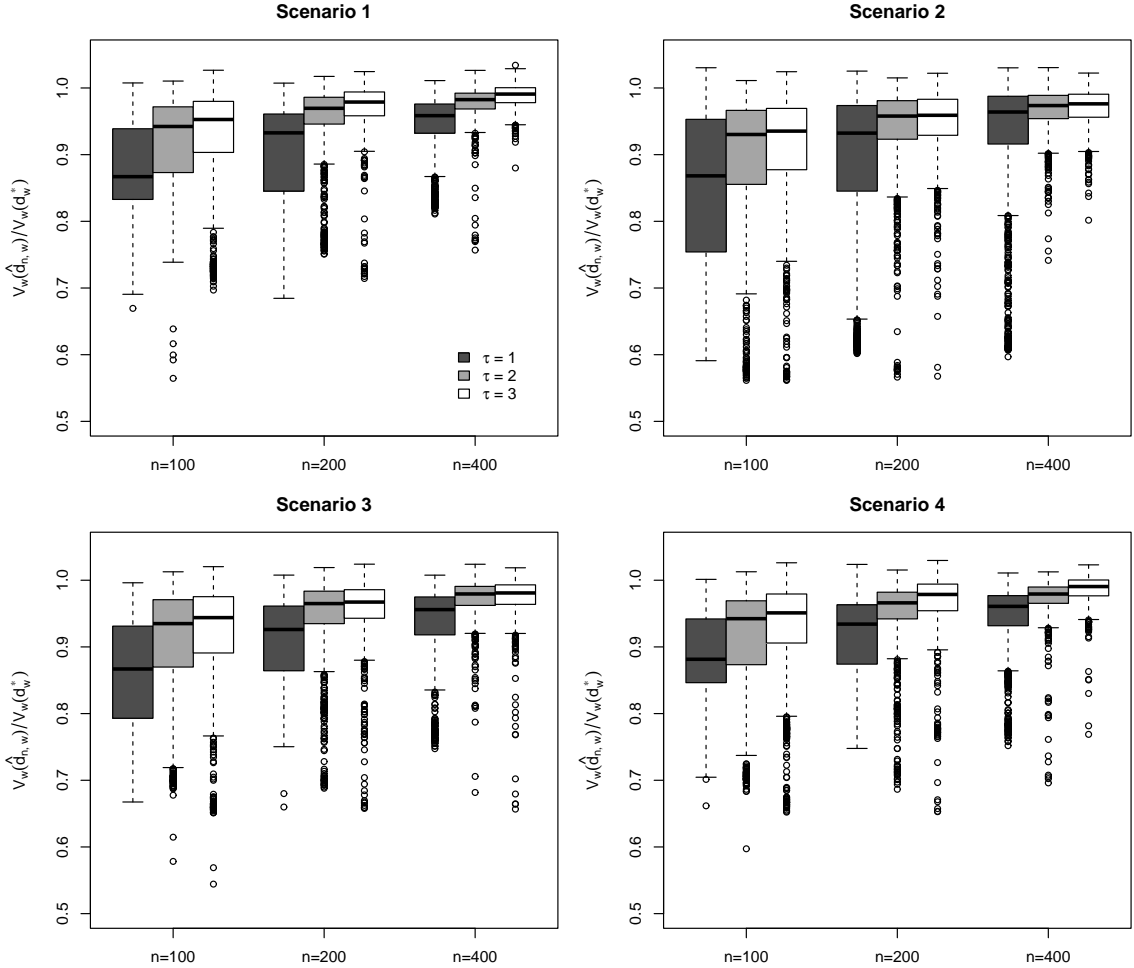}}
\caption{Simulation study: Performance of $\hat{d}_{n,w}$ for the duration of tumor response (i.e., $w=(0,1,0)'$), in terms of the estimated individualized treatment rule value ratio $\mathcal{V}_w(\hat{d}_{n,w})/\mathcal{V}_w(d_w^*)$ for different choices of the $\tau$ used in the analysis. The total length of the follow-up is equal to 3. Results under an average censoring rate of 59.5\%.}
\label{f:sim_taus_high}
\end{figure}

Simulation results on the performance of the proposed ITR estimator when $\mathcal{F}$ is the RKHS with the Gaussian kernel with $\sigma=1$ (less flexible kernel) and $\sigma=5$ (more flexible kernel), for the duration of tumor response, are depicted in Figures~\ref{f:resp_RKHS_low}--\ref{f:resp_RKHS_high}. For comparison, these figures also illustrate the performance of the proposed method when $\mathcal{F}$ is the space of linear functions. 

\begin{figure}
\centerline{\includegraphics[width=6in]{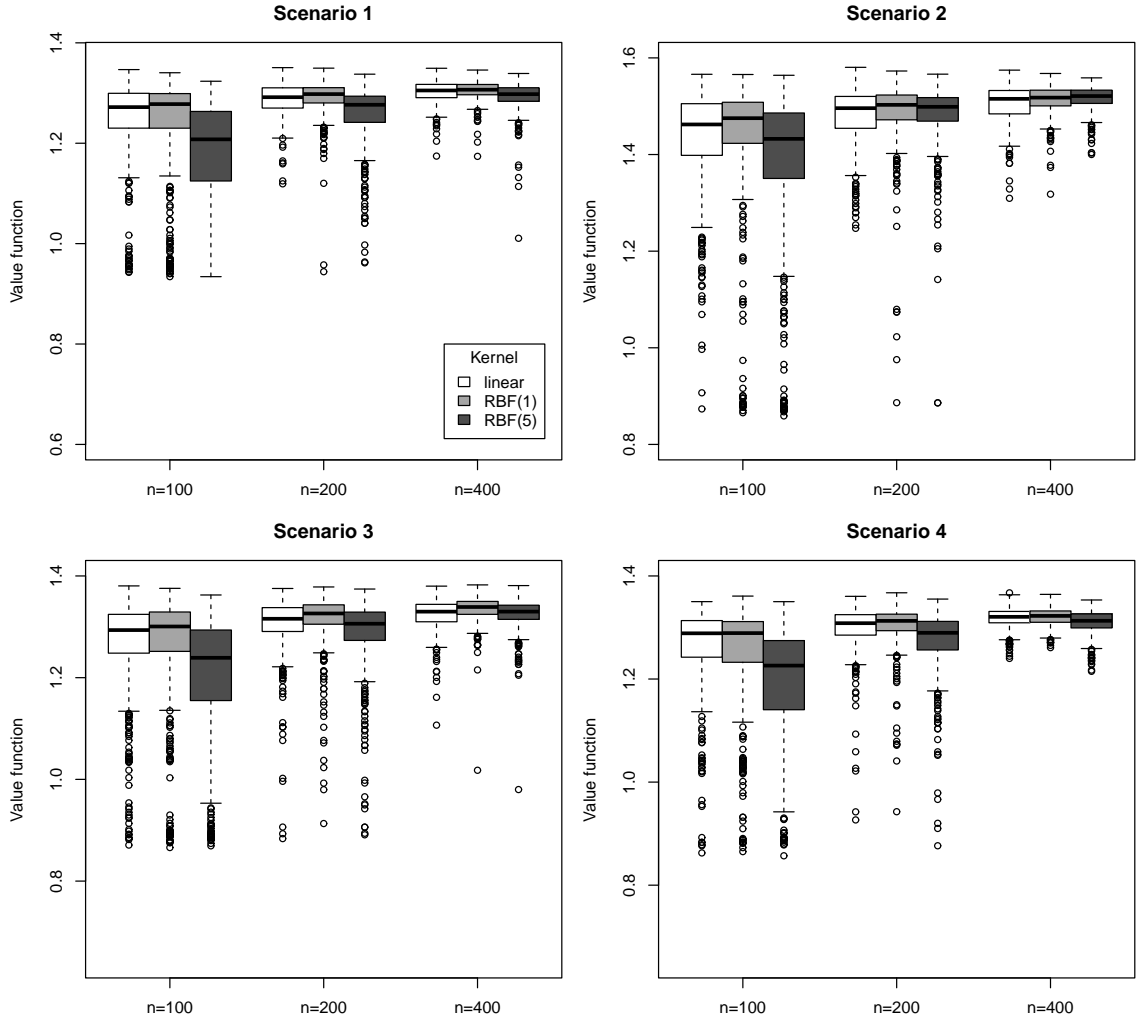}}
\caption{Simulation study: Value functions for the estimated individualized treatment rules for the duration of tumor response (i.e., $w=(0,1,0)'$), based on the proposed method with $\mathcal{F}$ being the class of linear functions or the RKHS with the Gaussian (also known as radial basis function) kernel with $\sigma=1$ (RBF(1); less flexible kernel) and $\sigma=5$ (RBF(5); more flexible kernel). Results under an average censoring rate of 28.4\%.}
\label{f:resp_RKHS_low}
\end{figure}

\begin{figure}
\centerline{\includegraphics[width=6in]{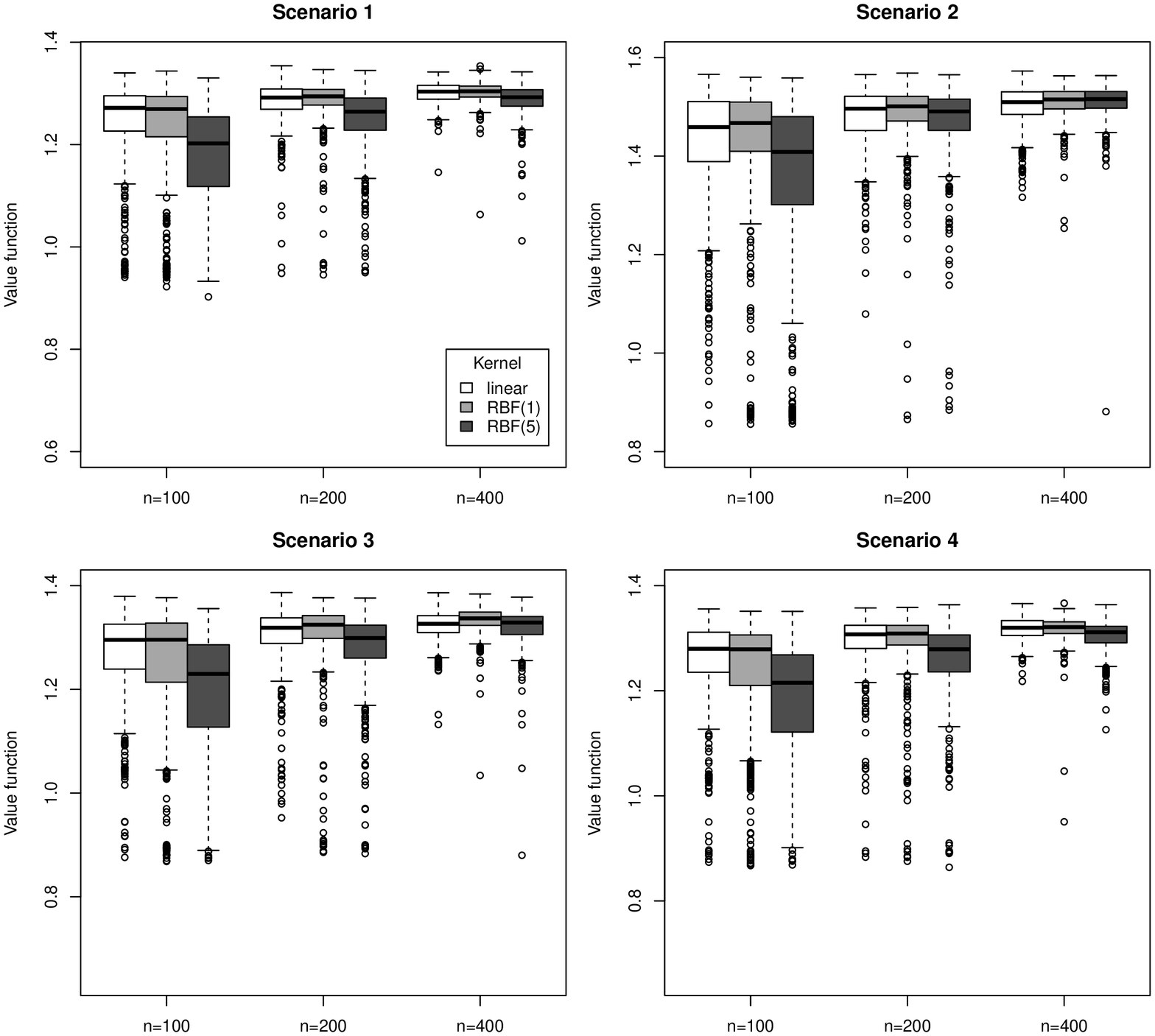}}
\caption{Simulation study: Value functions for the estimated individualized treatment rules for the duration of tumor response (i.e., $w=(0,1,0)'$), based on the proposed method with $\mathcal{F}$ being the class of linear functions or the RKHS with the Gaussian (also known as radial basis function) kernel with $\sigma=1$ (RBF(1); less flexible kernel) and $\sigma=5$ (RBF(5); more flexible kernel). Results under an average censoring rate of 42.8\%.}
\label{f:resp_RKHS_med}
\end{figure}

\begin{figure}
\centerline{\includegraphics[width=6in]{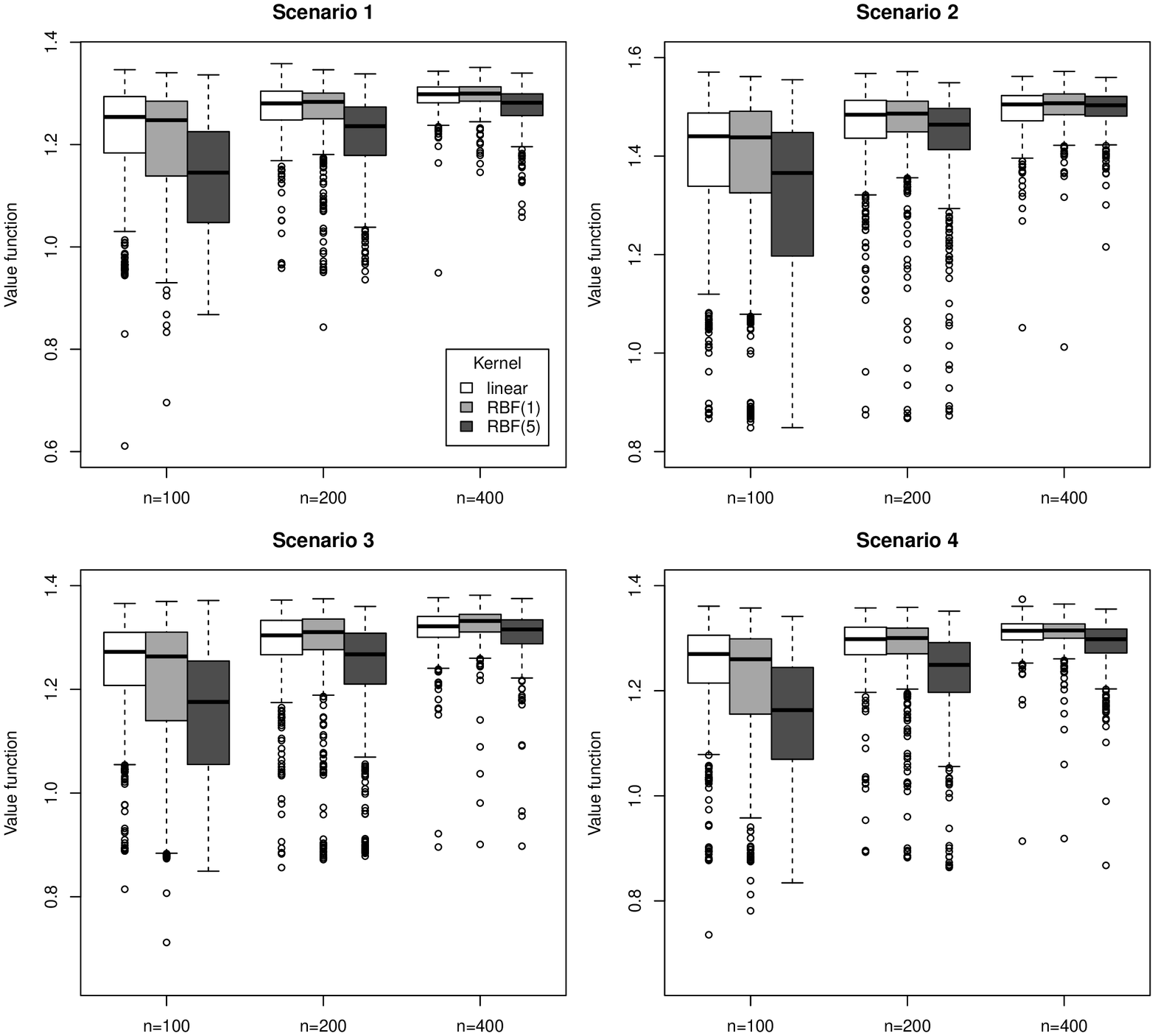}}
\caption{Simulation study: Value functions for the estimated individualized treatment rules for the duration of tumor response (i.e., $w=(0,1,0)'$), based on the proposed method with $\mathcal{F}$ being the class of linear functions or the RKHS with the Gaussian (also known as radial basis function) kernel with $\sigma=1$ (RBF(1); less flexible kernel) and $\sigma=5$ (RBF(5); more flexible kernel). Results under an average censoring rate of 59.5\%.}
\label{f:resp_RKHS_high}
\end{figure}

\subsection{Simulation results under $w=(1,1,0)'$ (progression-free survival time)}

Simulation results regarding the performance of the estimated ITR $\hat{d}_{n,w}$ for the progression-free survival time (i.e., $w=(1,1,0)'$) are depicted in Figures~\ref{f:surv_Zhao15_low}--\ref{f:surv_Zhao15_high}. For comparison, these plots also illustrate the performance of the inverse censoring weighted outcome weighted learning (ICO) and doubly robust outcome weighted learning (DR) methods by \citet{Zhao15} for censored failure times. 

\begin{figure}
\centerline{\includegraphics[width=6in]{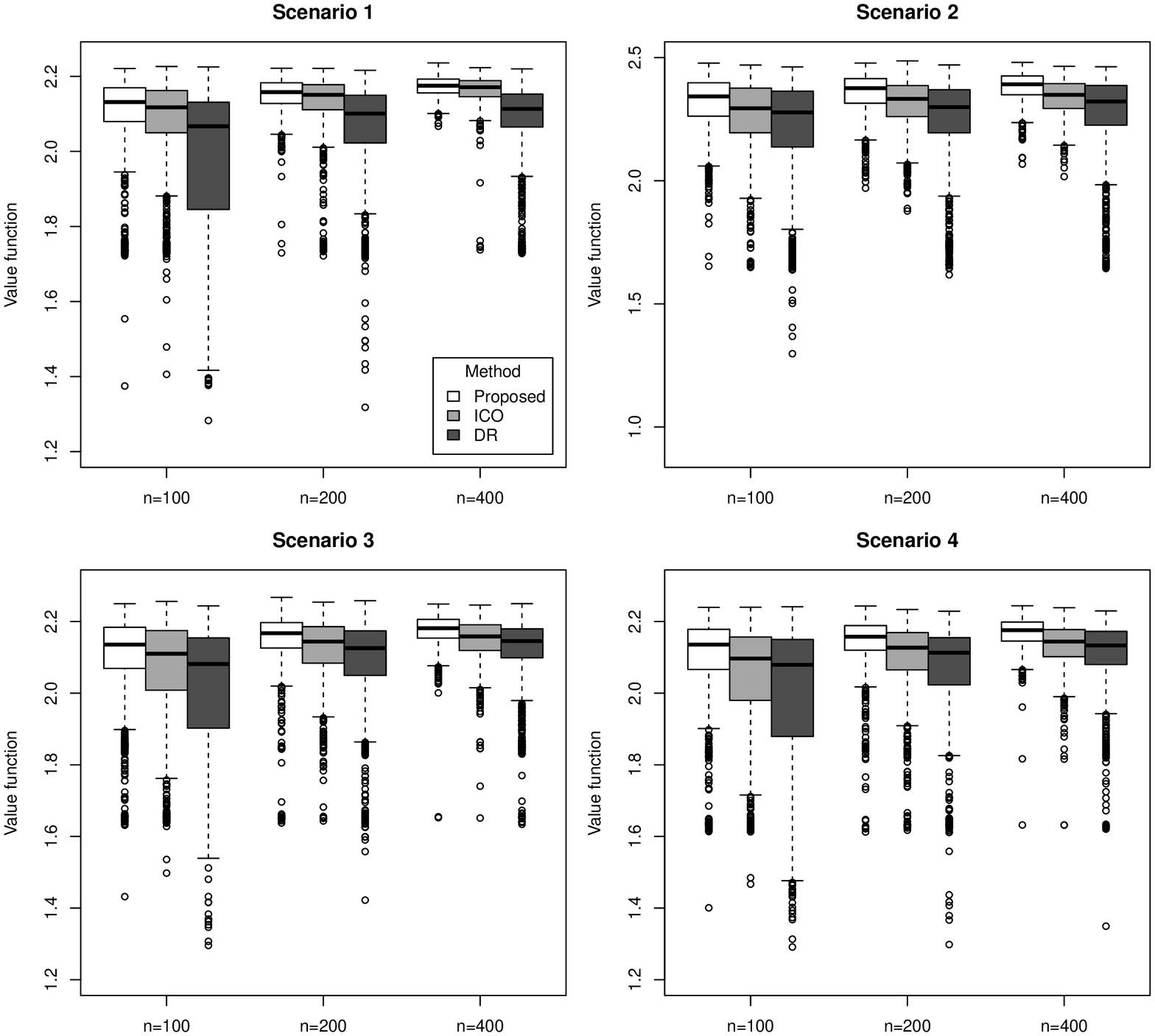}}
\caption{Simulation study: Value functions for the estimated individualized treatment rules based on the proposed method and the inverse censoring weighted outcome weighted learning (ICO) and doubly robust outcome weighted learning (DR) methods by \citet{Zhao15} for the progression-free survival time (i.e., $w=(1,1,0)'$). Results under an average censoring rate of 28.4\%.}
\label{f:surv_Zhao15_low}
\end{figure}

\begin{figure}
\centerline{\includegraphics[width=6in]{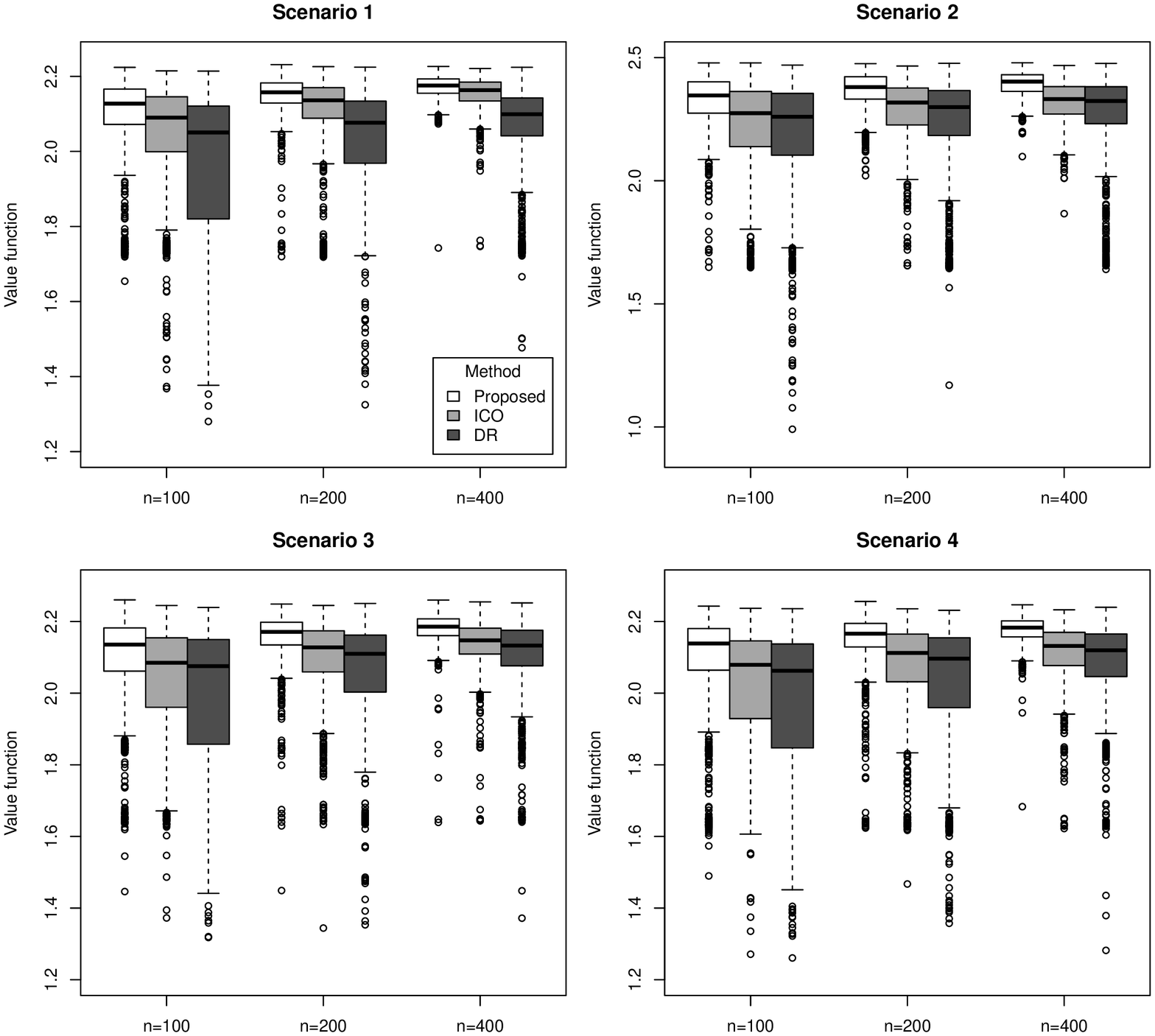}}
\caption{Simulation study: Value functions for the estimated individualized treatment rules based on the proposed method and the inverse censoring weighted outcome weighted learning (ICO) and doubly robust outcome weighted learning (DR) methods by \citet{Zhao15} for the progression-free survival time (i.e., $w=(1,1,0)'$). Results under an average censoring rate of 42.8\%.}
\label{f:surv_Zhao15_med}
\end{figure}

\begin{figure}
\centerline{\includegraphics[width=6in]{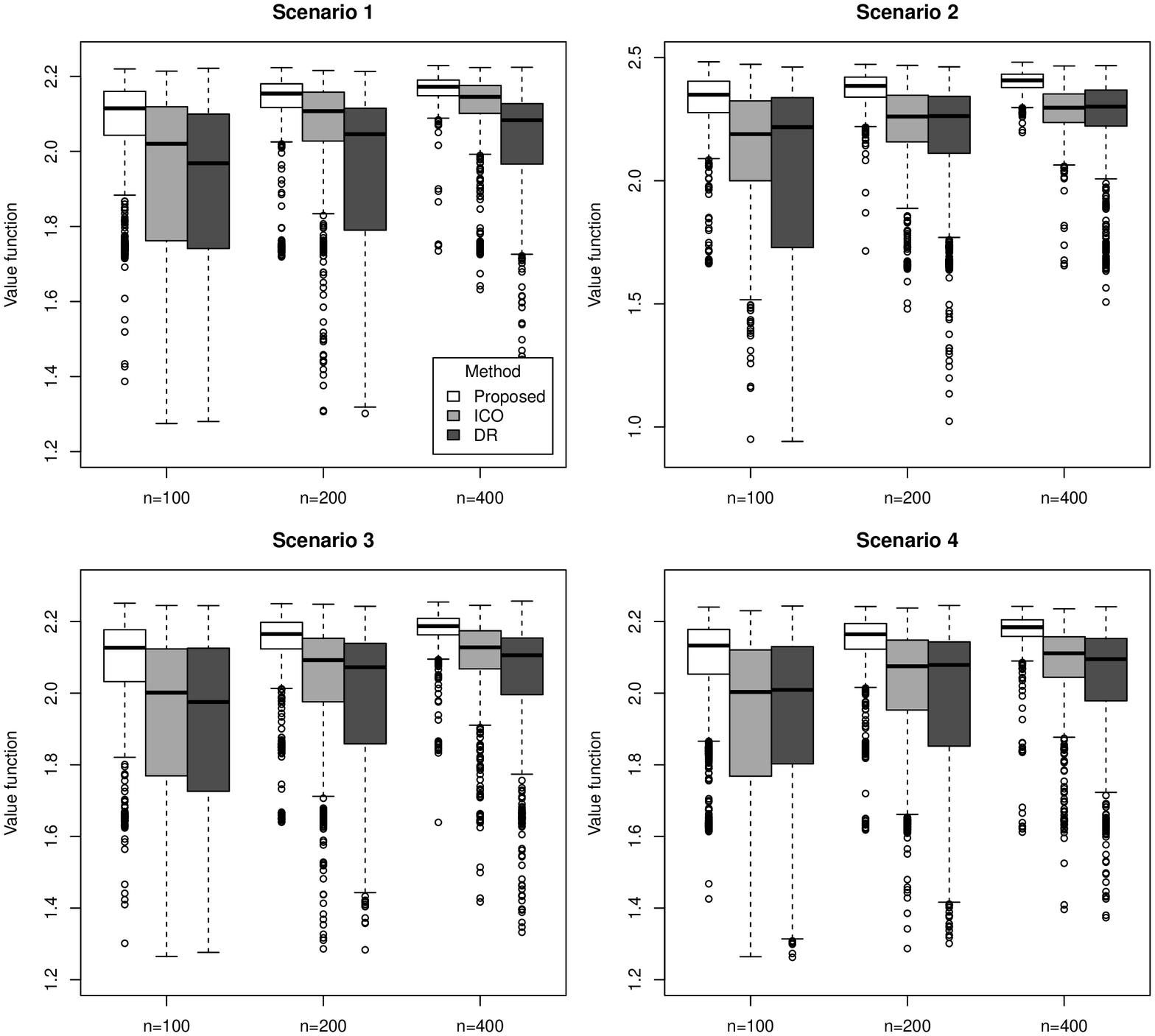}}
\caption{Simulation study: Value functions for the estimated individualized treatment rules based on the proposed method and the inverse censoring weighted outcome weighted learning (ICO) and doubly robust outcome weighted learning (DR) methods by \citet{Zhao15} for the progression-free survival time (i.e., $w=(1,1,0)'$). Results under an average censoring rate of 59.5\%.}
\label{f:surv_Zhao15_high}
\end{figure}

The simulation results regarding the validity of the proposed inference methods for $\mathcal{V}_w(\hat{d}_{n,w})$ in terms of the progression-free survival time are summarized in Tables~\ref{t:sims_lin} and \ref{t:sims_nonlin}.

\begin{table}
\caption{Simulation study: Performance of the proposed inference methods for the true value of the estimated ITR $\mathcal{V}_w(\hat{d}_{n,w})$ for the progression-free survival time (i.e., $w=(1,1,0)'$), under a linear optimal decision function $f_w^*$ (scenarios 1 and 2). (Cens: right censoring rate; $n$: training sample size; MCSD: Monte Carlo standard deviation of the estimates; ASE: Average of the standard error estimates; CP: empirical coverage probability of the 95\% confidence interval)}
\label{t:sims_lin}
\begin{center}
\begin{tabular}{lllcccccccc}
\hline
&&& \multicolumn{4}{c}{$\hat{\mathcal{V}}_{n,w}(\hat{d}_{n,w})$} & \multicolumn{4}{c}{$\hat{\mathcal{V}}_{n,w}^{jk}(\hat{d}_{n,w})$} \\
\cmidrule(lr){4-7}\cmidrule(lr){8-11}
Scenario & Cens & $n$ &
\multicolumn{1}{c}{\% error} &
\multicolumn{1}{c}{MCSD} & 
\multicolumn{1}{c}{ASE} & 
\multicolumn{1}{c}{CP} &
\multicolumn{1}{c}{\% error} &
\multicolumn{1}{c}{MCSD} & 
\multicolumn{1}{c}{ASE} & 
\multicolumn{1}{c}{CP} \\ \hline
1&  28\%&100&3.453&0.304&0.272&0.945&-2.439&0.368&0.272&0.893 \\
&&200&2.068&0.208&0.199&0.942&-0.779&0.218&0.199&0.935 \\
&&400&1.327&0.140&0.142&0.956&-0.502&0.147&0.142&0.940 \\[1ex]
&  43\%&100&4.397&0.318&0.300&0.956&-2.285&0.383&0.300&0.902 \\
&&200&3.009&0.221&0.218&0.952&-0.393&0.244&0.218&0.931 \\
&&400&1.627&0.151&0.155&0.959&-0.267&0.156&0.155&0.955 \\[1ex]
&  60\%&100&5.171&0.382&0.363&0.946&-4.364&0.476&0.363&0.886 \\
&&200&3.220&0.264&0.262&0.954&-1.236&0.315&0.262&0.914 \\
&&400&1.627&0.181&0.187&0.954&-0.856&0.200&0.187&0.939 \\[1.5ex]
2&  28\%&100&2.023&0.324&0.289&0.953&-0.910&0.331&0.289&0.941 \\
&&200&1.407&0.227&0.207&0.949&-0.270&0.228&0.207&0.938 \\
&&400&0.625&0.161&0.147&0.950&-0.249&0.162&0.147&0.944 \\[1ex]
&  42\%&100&2.472&0.344&0.319&0.948&-0.968&0.351&0.319&0.931 \\
&&200&1.839&0.236&0.227&0.947&-0.115&0.242&0.227&0.941 \\
&&400&0.901&0.171&0.161&0.941&-0.195&0.174&0.161&0.936 \\[1ex]
&  57\%&100&1.892&0.402&0.383&0.948&-3.099&0.438&0.383&0.919 \\
&&200&2.002&0.261&0.274&0.962&-0.678&0.272&0.274&0.944 \\
&&400&0.763&0.184&0.192&0.970&-0.741&0.188&0.192&0.954 \\
\hline
\end{tabular}
\end{center}
\end{table}

\begin{table}
\caption{Simulation study: Performance of the proposed inference methods for the true value of the estimated individualized treatment rule $\mathcal{V}_w(\hat{d}_{n,w})$ for the progression-free survival time (i.e., $w=(1,1,0)'$), under a nonlinear optimal decision function $f_w^*$ (scenarios 3 and 4). (Cens: right censoring rate; $n$: training sample size; MCSD: Monte Carlo standard deviation of the estimates; ASE: Average of the standard error estimates; CP: empirical coverage probability of the 95\% confidence interval)}
\label{t:sims_nonlin}
\begin{center}
\begin{tabular}{lllcccccccc}
\hline
&&& \multicolumn{4}{c}{$\hat{\mathcal{V}}_{n,w}(\hat{d}_{n,w})$} & \multicolumn{4}{c}{$\hat{\mathcal{V}}_{n,w}^{jk}(\hat{d}_{n,w})$} \\
\cmidrule(lr){4-7}\cmidrule(lr){8-11}
Scenario & Cens & $n$ &
\multicolumn{1}{c}{\% error} &
\multicolumn{1}{c}{MCSD} & 
\multicolumn{1}{c}{ASE} & 
\multicolumn{1}{c}{CP} &
\multicolumn{1}{c}{\% error} &
\multicolumn{1}{c}{MCSD} & 
\multicolumn{1}{c}{ASE} & 
\multicolumn{1}{c}{CP} \\ \hline
3&  29\%&100&3.511&0.313&0.271&0.958&-1.521&0.361&0.271&0.908 \\
&&200&1.992&0.214&0.198&0.959&-0.967&0.234&0.198&0.937 \\
&&400&1.012&0.152&0.142&0.957&-0.573&0.156&0.142&0.942 \\[1ex]
&  43\%&100&4.064&0.344&0.298&0.953&-2.430&0.401&0.298&0.905 \\
&&200&2.574&0.239&0.217&0.938&-0.673&0.266&0.217&0.914 \\
&&400&1.379&0.156&0.155&0.957&-0.461&0.162&0.155&0.947 \\[1ex]
&  60\%&100&5.864&0.389&0.362&0.947&-3.478&0.466&0.362&0.894 \\
&&200&3.404&0.272&0.261&0.962&-0.944&0.316&0.261&0.930 \\
&&400&1.826&0.179&0.186&0.961&-0.638&0.197&0.186&0.954 \\[1.5ex]
4&  29\%&100&3.767&0.323&0.271&0.930&-1.907&0.380&0.271&0.887 \\
&&200&2.016&0.222&0.198&0.941&-0.895&0.240&0.198&0.920 \\
&&400&1.419&0.146&0.142&0.953&-0.107&0.150&0.142&0.946 \\[1ex]
&  43\%&100&4.548&0.321&0.298&0.951&-1.714&0.389&0.298&0.901 \\
&&200&2.319&0.232&0.216&0.954&-0.840&0.257&0.216&0.925 \\
&&400&1.475&0.159&0.155&0.956&-0.184&0.161&0.155&0.945 \\[1ex]
&  59\%&100&5.520&0.408&0.364&0.947&-3.631&0.494&0.364&0.883 \\
&&200&3.142&0.270&0.262&0.958&-0.975&0.295&0.262&0.933 \\
&&400&2.067&0.189&0.187&0.960&-0.180&0.197&0.187&0.949 \\
\hline
\end{tabular}
\end{center}
\end{table}

Simulation results on the performance of the proposed ITR estimator when $\mathcal{F}$ is the RKHS with the Gaussian kernel with $\sigma=1$ (less flexible kernel) and $\sigma=5$ (more flexible kernel), for the duration of tumor response, are depicted in Figures~\ref{f:surv_RKHS_low}--\ref{f:surv_RKHS_high}. For comparison, these figures also illustrate the performance of the proposed method when $\mathcal{F}$ is the space of linear functions. 

\begin{figure}
\centerline{\includegraphics[width=6in]{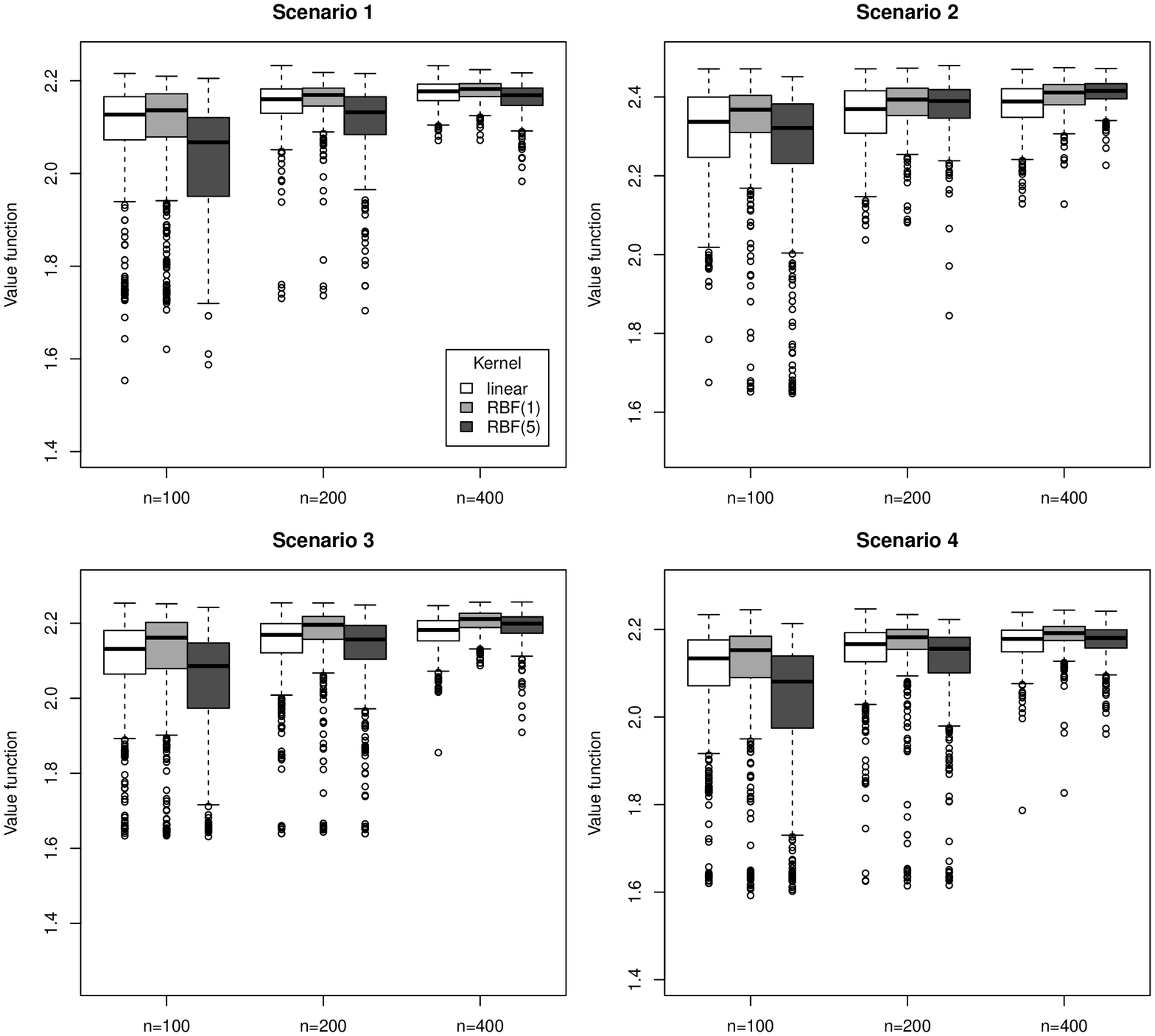}}
\caption{Simulation study: Value functions for the estimated individualized treatment rules for the progression-free survival time (i.e., $w=(1,1,0)'$), based on the proposed method with $\mathcal{F}$ being the class of linear functions or the RKHS with the Gaussian (also known as radial basis function) kernel with $\sigma=1$ (RBF(1); less flexible kernel) and $\sigma=5$ (RBF(5); more flexible kernel). Results under an average censoring rate of 28.4\%.}
\label{f:surv_RKHS_low}
\end{figure}

\begin{figure}
\centerline{\includegraphics[width=6in]{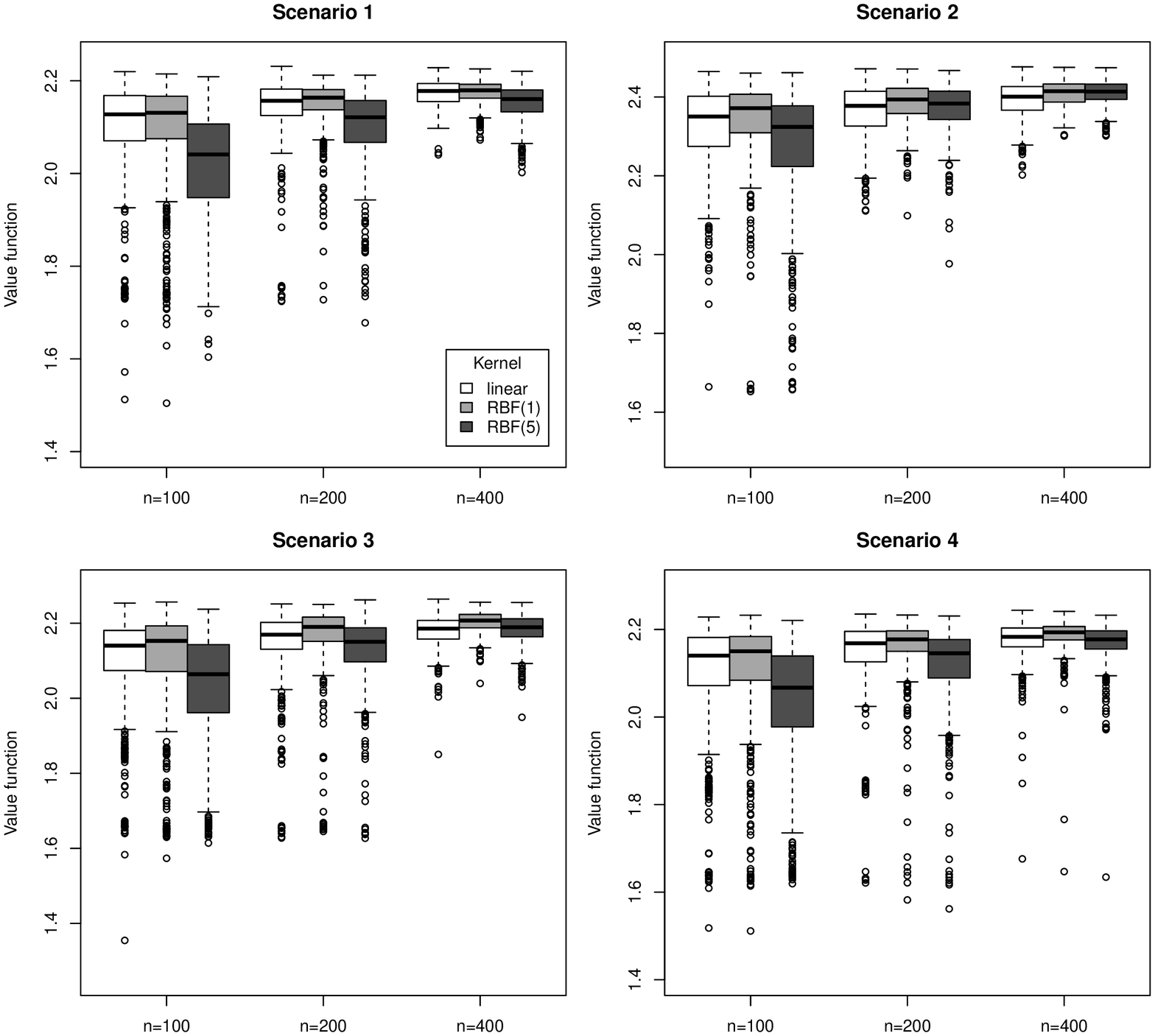}}
\caption{Simulation study: Value functions for the estimated individualized treatment rules for the progression-free survival time (i.e., $w=(1,1,0)'$), based on the proposed method with $\mathcal{F}$ being the class of linear functions or the RKHS with the Gaussian (also known as radial basis function) kernel with $\sigma=1$ (RBF(1); less flexible kernel) and $\sigma=5$ (RBF(5); more flexible kernel). Results under an average censoring rate of 42.8\%.}
\label{f:surv_RKHS_med}
\end{figure}

\begin{figure}
\centerline{\includegraphics[width=6in]{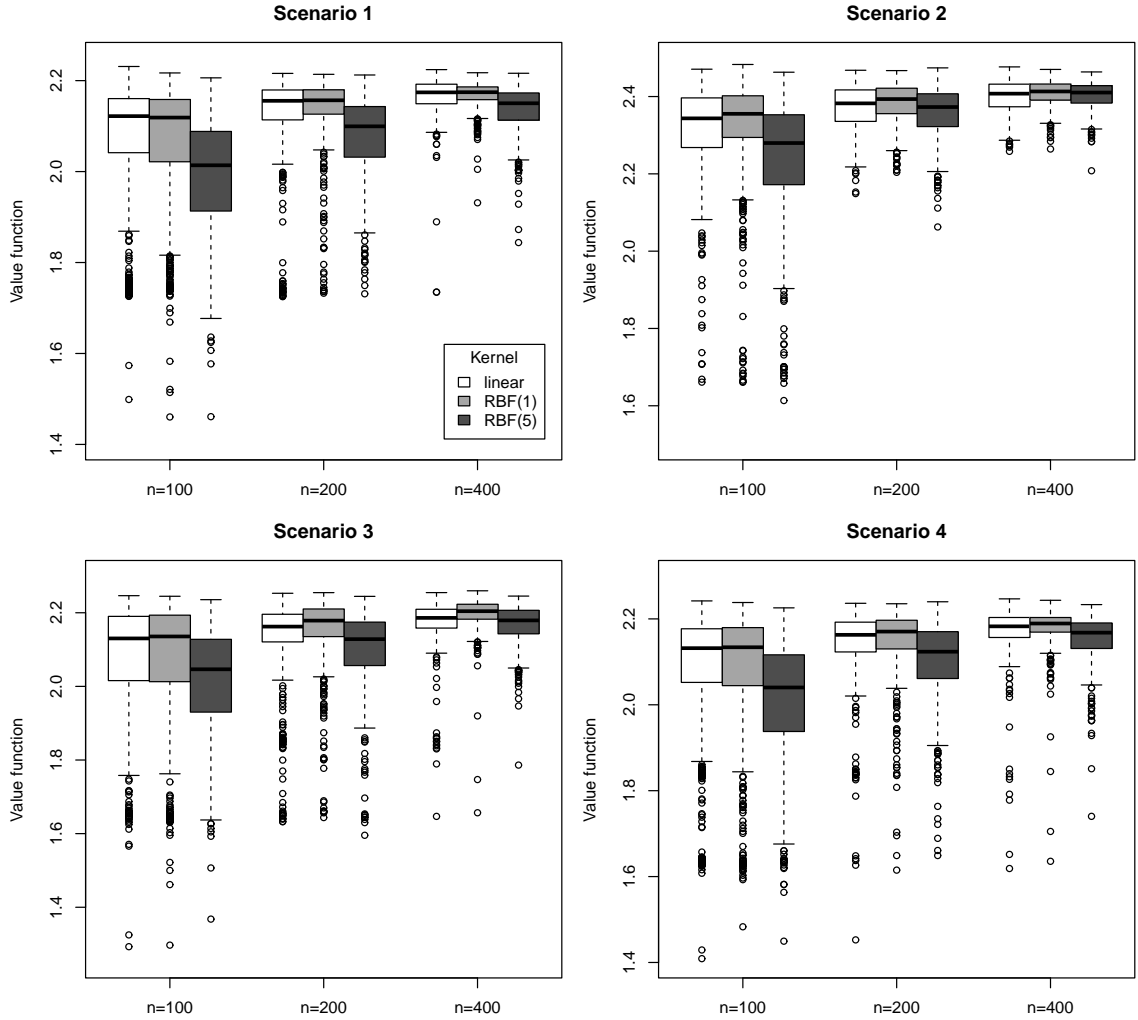}}
\caption{Simulation study: Value functions for the estimated individualized treatment rules for the progression-free survival time (i.e., $w=(1,1,0)'$), based on the proposed method with $\mathcal{F}$ being the class of linear functions or the RKHS with the Gaussian (also known as radial basis function) kernel with $\sigma=1$ (RBF(1); less flexible kernel) and $\sigma=5$ (RBF(5); more flexible kernel). Results under an average censoring rate of 59.5\%.}
\label{f:surv_RKHS_high}
\end{figure}

\section*{Appendix D. Additional Results from the SPECTRUM Trial Analysis}
\renewcommand{\thesubsection}{D.\arabic{subsection}}
\setcounter{subsection}{0}

The estimates of the treatment-specific cumulative transition intensities and state occupation probabilities based on the data from the SPECTRUM trial are depicted in Figures~\ref{f:cti} and \ref{f:sop}.

\begin{figure}
\centerline{\includegraphics[width=6in]{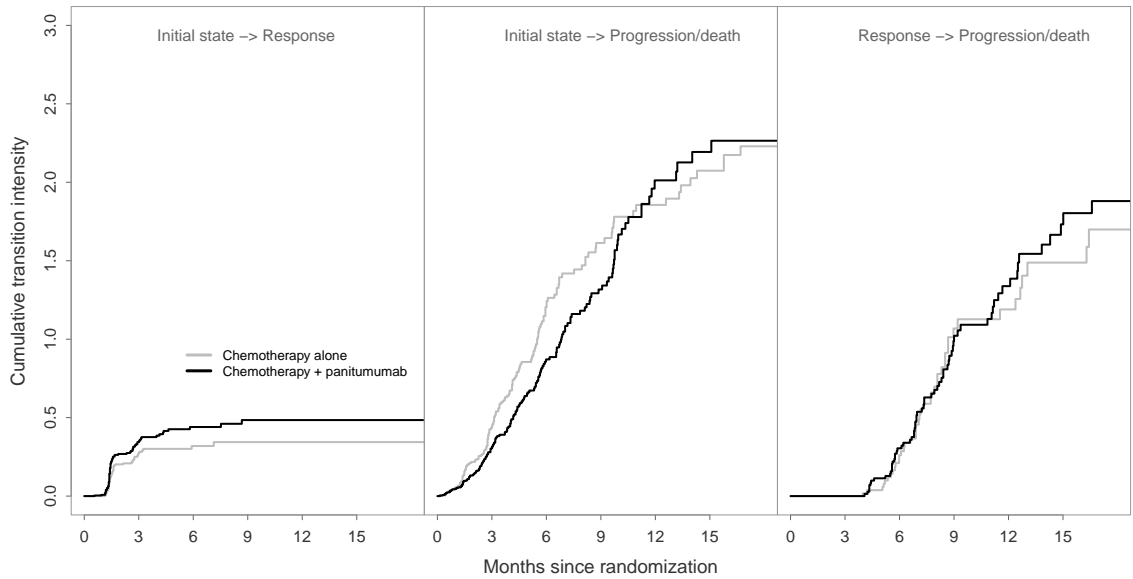}}
\caption{SPECTRUM trial analysis: Nonparametric estimates of the cumulative transition intensities by treatment arm.}
\label{f:cti}
\end{figure}

\begin{figure}
\centerline{\includegraphics[width=6in]{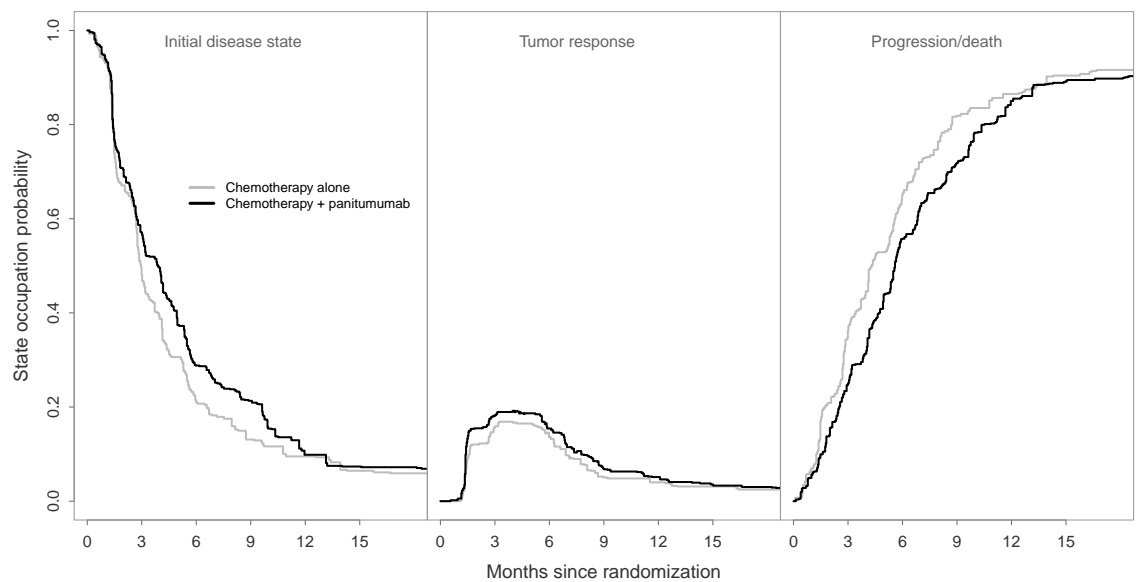}}
\caption{SPECTRUM trial analysis: Nonparametric estimates of the state occupation probabilities by treatment arm.}
\label{f:sop}
\end{figure}

\section*{Appendix E. Relaxing the independent censoring assumption}
\renewcommand{\thesubsection}{E.\arabic{subsection}}
\setcounter{subsection}{0}

A plausible relaxation of the independent censoring assumption (condition C1) is to allow censoring to depend on treatment $A$, since censoring will likely be higher among those receiving the treatment with the greater toxicity \citep{Templeton20}. This can be trivially incorporated into the proposed methodology by simply replacing $\hat{\Lambda}_n(\tilde{T}\wedge t)$ in the proposed estimators with the nonparametric estimate of the conditional cumulative hazard of censoring given $A$, that is
\[
\hat{\Lambda}_n(\tilde{T}\wedge t,A)=\hat{\Lambda}_{n,1}(\tilde{T}\wedge t)I(A=1) + \hat{\Lambda}_{n,-1}(\tilde{T}\wedge t)I(A=-1),
\]
where 
\[
\hat{\Lambda}_{n,a}(t)=\int_0^t\frac{\sum_{i=1}^nI(A_i=a)\textrm{d}N_i(s)}{\sum_{i=1}^nI(A_i=a)Y_i(s)}, \ \ \ \ a\in\{-1, 1\}, \ \ t\in[0,\tau].
\]
This estimator is uniformly consistent for the true conditional cumulative hazard of censoring given $A$
\[
\Lambda_0(\tilde{T}\wedge t,A)=\Lambda_{0,1}(\tilde{T}\wedge t)I(A=1) + \Lambda_{0,-1}(\tilde{T}\wedge t)I(A=-1),
\]
where $\Lambda_{0,a}(t)$, $a\in\{-1,1\}$, is the true cumulative hazard of censoring among those with $A=a$. In this case, all the theoretical poperties of the proposed estimators in the main manuscript still hold, with the exception that the influence functions $\psi_i(f)$ have the (slightly different) form
\begin{eqnarray*}
\psi_{i,w}(f)&=&\int_0^{\tau}\frac{Y_{i,w}(t)I(C_i\geq T_i\wedge t)I[A_i=\textrm{sgn}\{f(Z_i)\}]}{\exp\{-\Lambda_0(\tilde{T}_i\wedge t,A_i)\}\{A_i\pi_0 + (1-A_i)/2\}}\textrm{d}m(t)-\mathcal{V}_w(\textrm{sgn}(f)) \\
&& -E\left\{A\int_0^{\tau}\frac{Y_{w}(t)I(C\geq T\wedge t)I[A=\textrm{sgn}\{f(Z)\}]}{\exp\{-\Lambda_0(\tilde{T}\wedge t,A)\}\{A\pi_0 + (1-A)/2\}^2}\textrm{d}m(t)\right\}\{I(A_i=1)-\pi_0\} \\
&&+\eta_{\Lambda_{0},f,w}'(\gamma_i), \ \ \ \ i=1,\ldots,n, \ \ w\in\mathcal{W}, \ \ f\in\mathcal{F},
\end{eqnarray*}
where
\[
\mathcal{V}_w(d)=E\left(\left[\int_0^{\tau}\frac{Y_{w}(t)I(C\geq T\wedge t)}{\exp\{-\Lambda_0(\tilde{T}\wedge t,A)\}}\textrm{d}m(t)\right]\frac{I(A=d(Z))}{A\pi_0 + (1-A)/2}\right), 
\]
\[
\eta_{\Lambda_0,f,w}'(h)=E\left\{\int_0^{\tau}\frac{Y_{w}(t)I(C\geq T\wedge t)I[A=\textrm{sgn}\{f(Z)\}]}{\exp\{-\Lambda_0(\tilde{T}\wedge t,A)\}\{A\pi_0 + (1-A)/2\}}h(\tilde{T}\wedge t,A)\textrm{d}m(t)\right\},
\]
for $h$ in the space $D[0,\tau]$ of right continuous functions on $[0,\tau]$ with left hand limits and
\[
\gamma_i(t,a)=\sum_{j\in\{-1,1\}}I(a=j)\left[\int_0^t\frac{I(A_i=j)dN_i(s)}{E\{I(A=j)Y(s)\}}-\int_0^t\frac{I(A_i=j)Y_i(s)}{E\{I(A=j)Y(s)\}}\textrm{d}\Lambda_{0,a}(s)\right], \ \ \ \ t\in[0,\tau],
\]
for $i=1,\ldots,n$. The corresponding empirical versions of the latter influence functions are
\begin{eqnarray*}
\hat{\psi}_{i,w}(f)&=&\int_0^{\tau}\frac{Y_{i,w}(t)I(C_i>T_i\wedge t)I[A_i=\textrm{sgn}\{f(Z_i)\}]}{\exp\{-\hat{\Lambda}_n(\tilde{T}_i\wedge t,A_i)\}\{A_i\hat{\pi}_n + (1-A_i)/2\}}\textrm{d}m(t)-\hat{\mathcal{V}}_{n,w}(\textrm{sgn}(f)) \\
&&-\left\{\frac{1}{n}\sum_{j=1}^nA_j\int_0^{\tau}\frac{Y_{j,w}(t)I(C_j>T_j\wedge t)I[A_j=\textrm{sgn}\{f(Z_j)\}]}{\exp\{-\hat{\Lambda}_n(\tilde{T}_j\wedge t,A_j)\}\{A_j\hat{\pi}_n + (1-A_j)/2\}^2}\textrm{d}m(t)\right\}\{I(A_i=1)-\hat{\pi}_n\} \\
&&+\hat{\eta}_{\hat{\Lambda}_n,f,w}'(\hat{\gamma}_i), \ \ \ \ i=1,\ldots,n, \ \ w\in\mathcal{W}, \ \ f\in\mathcal{F},
\end{eqnarray*}
where
\[
\hat{\mathcal{V}}_{n,w}(d)=\frac{1}{n}\sum_{i=1}^n\left(\left[\int_0^{\tau}\frac{Y_{i,w}(t)I(C_i\geq T_i\wedge t)}{\exp\{-\hat{\Lambda}_n(\tilde{T}_i\wedge t,A_i)\}}\textrm{d}m(t)\right]\frac{I(A_i=d(Z_i))}{A_i\hat{\pi}_n + (1-A_i)/2}\right),
\]
\[
\hat{\eta}_{\hat{\Lambda}_n,f,w}'(h)=\frac{1}{n}\sum_{i=1}^n\left\{\int_0^{\tau}\frac{Y_{i,w}(t)I(C_i>T_i\wedge t)I[A_i=\textrm{sgn}\{f(Z_i)\}]}{\exp\{-\hat{\Lambda}_n(\tilde{T}_i\wedge t,A_i)\}\{A_i\hat{\pi}_n + (1-A_i)/2\}}h(\tilde{T}_i\wedge t,A_i)\textrm{d}m(t)\right\},
\]
for $h\in D[0,\tau]$, and
\[
\hat{\gamma}_i(t,a)=\sum_{j\in\{-1,1\}}I(a=j)\left[\int_0^t\frac{I(A_i=j)\textrm{d}N_i(s)}{n^{-1}\sum_{l=1}^nI(A_l=j)Y_l(s)}-\int_0^t\frac{I(A_i=j)Y_i(s)}{n^{-1}\sum_{l=1}^nI(A_l=j)Y_l(s)}\textrm{d}\hat{\Lambda}_{n,a}(s)\right], 
\]
for $i=1,\ldots,n$ and $t\in[0,\tau]$.

A further relaxation of the independent censoring assumption is to allow censoring to depend on both $A$ and $Z$. In this case, one can impose a semiparametric Cox model of the form $\Lambda(t;A,Z)=\Lambda_0(t)\exp\{\theta'(A,Z')'\}$ for the right censoring time, and use the estimated conditional hazard in the proposed objective function and value function estimators. Provided that this model is correctly specified, the theoretical properties of the proposed method still hold, with the exception that $\gamma_i(t)$ in $\psi_{i,w}(f)$ (see Appendix B) is replaced by the influence function of $\sqrt{n}\{\hat{\Lambda}_n(t)\exp(\hat{\theta}_n'(a,z')')-\Lambda(t;a,z)\}$ under partial likelihood estimation.

\end{appendices}

\end{document}